\documentclass[10pt,letter,final]{IEEEtran}
\usepackage{amsthm,amssymb,graphicx,multirow,amsmath,color,amsfonts}
\usepackage[update,prepend]{epstopdf}%
\usepackage[noadjust]{cite}
\usepackage{tikz}
\usetikzlibrary{arrows,calc}		
\usepackage{bbm} 
\usepackage{pdfpages}
\usepackage{tabulary}
\usepackage{multirow}
\usepackage{comment}
\usepackage{dsfont}
\usepackage{pgf,tikz}
\usepackage{paralist}{
\usepackage{amsfonts}
\usepackage{setspace,lipsum}
\usepackage{url}
\usepackage{bigints}
\usepackage{mathtools}

\usepackage[english]{babel}
\usepackage[autostyle]{csquotes}
\usepackage{textcomp}
\usepackage{varwidth}
\usepackage{tcolorbox}
\usepackage{lipsum}

\let\OLDthebibliography\thebibliography
\renewcommand\thebibliography[1]{
  \OLDthebibliography{#1}
  \setlength{\parskip}{0pt}
  \setlength{\itemsep}{0pt plus 0.3ex}
}

\newcommand{\cdf}{\mathtt{CDF}}
\newcommand{\ccdf}{\mathtt{CCDF}}
\newcommand{\pdf}{\mathtt{pdf}}
\newcommand{\pmf}{\mathtt{pmf}}

\newcommand\numberthis{\addtocounter{equation}{1}\tag{\theequation}}

\newcommand{\x}{\mathbf{x}}
\newcommand{\y}{\mathbf{y}}

\makeatletter
\renewcommand{\env@cases}[1][@{}l@{\quad}l@{}]{%
  \let\@ifnextchar\new@ifnextchar
  \left\lbrace
  \def\arraystretch{1.2}%
  \array{#1}%
}
\makeatother

\def\nb0{{\mathbf{0}}}
\def\nb1{{\mathbf{1}}}

\def\ncalA{{\mathcal{A}}}
\def\ncalB{{\mathcal{B}}}
\def\ncalC{{\mathcal{C}}}
\def\ncalD{{\mathcal{D}}}
\def\ncalE{{\mathcal{E}}}
\def\ncalF{{\mathcal{F}}}

\def\ncalP{{\mathcal{P}}}

\def\ncalR{{\mathcal{R}}}

\def\ncalZ{{\mathcal{Z}}}

\def\nbbE{{\mathbb{E}}}

\def\nbbP{{\mathbb{P}}}

\def\nbbR{{\mathbb{R}}}

\newtheorem{lemma}{Lemma}

\newtheorem{definition}{Definition}

\newtheorem{theorem}{Theorem}

\newtheorem{cor}{Corollary}

\newtheorem{assumption}{Assumption}

\def\argmax{\operatorname{arg~max}}
\def\E{\mathbb{E}}

\def\P{\mathbb{P}}

\def\p{p}

\def\R{\mathbb{R}}

\def\sir{\mathtt{SIR}}

\def\csr{\mathtt{CSR}}

\def\ec{\mathtt{EC}}
\def\sec{\mathtt{SEC}}

\allowdisplaybreaks 

\begin{document}
\graphicspath{{./Figures/}}
\title{Downlink Analysis of NOMA-enabled Cellular Networks with 3GPP-inspired User Ranking}
\author{
Praful D. Mankar,~\IEEEmembership{Member,~IEEE,} and Harpreet S. Dhillon,~\IEEEmembership{Senior~Member,~IEEE}
\thanks{The authors are with Wireless@VT, Bradley Department of Electrical and Computer Engineering, Virginia Tech, Blacksburg, VA. Email: \{prafuldm, hdhillon\}@vt.edu. This paper was presented in part at  IEEE Globecom 2019 \cite{mankar2019meta}. The support of the US National Science Foundation (Grant CNS-1814477) is gratefully acknowledged.
}
}

\maketitle
\begin{abstract} 
This paper provides a comprehensive downlink analysis of non-orthogonal multiple access (NOMA) enabled cellular networks using tools from stochastic geometry. As a part of this analysis, we develop a novel 3GPP-inspired user ranking technique to construct a user cluster for the non-orthogonal transmission by grouping users from the cell center (CC) and cell edge (CE) regions. This technique {allows the pairing of users} with distinct link qualities, which is imperative for harnessing  NOMA performance gains. The analysis is performed from the perspective of the {\em typical cell}, which is significantly different from the standard stochastic geometry-based approach of analyzing the performance of the {\em typical user}.
For this setting, we first derive the moments of the meta distributions for the CC and CE users under NOMA and orthogonal multiple access (OMA). Using this, we then derive the distributions of the transmission rates and mean packet delays under non-real time (NRT) and real-time (RT) service models, respectively, for both CC and CE users. Finally, we study two resource allocation (RA) techniques {to maximize the cell sum-rate} ($\csr$) under NRT service, and the sum effective capacity ($\sec$) under RT service. In addition to providing several useful design insights, our results demonstrate that  NOMA provides improved rate region and higher $\csr$ as compared to OMA. In addition, we also show that  NOMA provides better $\sec$ as compared to OMA for the higher user density.
\end{abstract}
\section{Introduction}
\label{sec:Introduction}
Owing to its superior  spectral efficiency,  {non-orthogonal multiple access (NOMA)} technique has emerged as a promising candidate for future wireless cellular networks.  Unlike traditional  {orthogonal multiple access (OMA)}, NOMA enables the  {base stations (BSs)} to concurrently serve multiple users using the same resource block (RB); see \cite{ding2017survey} and the references therein. 
In NOMA, the BS superimposes multiple layers of messages at different power levels and the user decodes its intended message using successive interference cancellation  (SIC) technique. In particular, the users with weaker channel qualities are assigned with higher powers so that their intra-cell interference is smaller.  For such power allocation, the user first  decodes and cancels interference power from the layers assigned to the users with weaker channels successively using SIC and then decodes its intended message.   {A key design aspect for NOMA is the pairing of users for non-orthogonal transmission. For example, consider  there is a user with good channel quality,  say UE$_1$, and a user with poor channel quality, say UE$_2$. In OMA, the associated BS ends up  allocating most of the time slots to UE$_2$ in order to ensure its quality of service (QoS) requirement which leads to lesser transmission opportunities for UE$_1$ and consequently smaller cell throughput. On the other hand, NOMA  allows the BS to concurrently serve both UE$_1$ and UE$_2$ using the same spectral resource, which provides uninterrupted channel access to UE$_1$ while meeting the required QoS constraints of UE$_2$. Although the UE$_1$'s {\em per-slot} throughput may suffer (compared to OMA) because of the need to decode both the messages (corresponding to UE$_1$ and UE$_2$), the overall transmission rate may still improve because of more transmission opportunities compared to OMA. As a result, NOMA can result in higher cell  ergodic capacity as compared to OMA \cite{Ding_NOMA_2014}. In this regard, it is beneficial to pair users with distinct link qualities. The readers can refer to \cite{ding2016impact} for more details on the effects of user pairing on NOMA performance.}
\subsection{Prior Art}
\label{subsec:PriorArt}
 The set of users scheduled for the non-orthogonal transmission in NOMA is often termed as the {\em user cluster}. To form a user cluster, one needs to first rank  users based on their channel gains (both large-scale path losses and small-scale fading effects) and perceived inter-cell interference, and then place users with distinct link qualities in the same cluster \cite{ding2016impact}.  However, since the  prior works on NOMA  are mostly focused on the {single-cell} NOMA analysis, the users are ranked solely based on  their distances from the BS \cite{choi2016power,liu2016cooperative} or based on their link qualities \cite{Ding_NOMA_2014,ding2015cooperative,timotheou2015fairness,zhu2017optimal}.    While such user ranking is meaningful in the context of {single-cell} analysis, it ignores  the impact of the {inter-cell} interference which is crucial for accurate performance analysis of NOMA  \cite{ali2017non}. 
{Recently, stochastic geometry has emerged as a popular choice for the analysis of large scale cellular networks \cite{haenggi2012stochastic,AndrewsGD16,blaszczyszyn_haenggi_keeler_mukherjee_2018}. However, incorporating sophisticated user ranking jointly based on the link qualities and inter-cell interference is challenging.} 
  This is because of the correlation in the corresponding desired and inter-cell interference powers received by the users within the same cell. Therefore, most of the existing works in this direction ignore these correlations and instead rank the users in the order of their mean desired signal powers (i.e., link distances) so that the $i$-th closest user becomes the $i$-th strongest user. 
Within this general direction, the  authors of \cite{ali2019downlink,ali2019downlinkLetter,salehi2018meta,salehi2018accuracy,tabassum2017modeling} analyzed $N$-ranked NOMA in cellular networks assuming that the BSs  follow a PPP. In \cite{tabassum2017modeling}, the uplink success probability is derived assuming that  the users follow a Poisson cluster process (PCP). In \cite{ali2019downlink}, the downlink success probability is derived while forming the user cluster within the indisk of the Poisson-Voronoi (PV) cell. However, this may underestimate the NOMA performance gains because users within the indisk of a PV cell will usually experience similar channel conditions and hence lack channel gain imbalance that results in the NOMA gains (see \cite{ding2016impact}).   
{By ranking the users based on their link distances, the authors of \cite{salehi2018meta,ali2019downlinkLetter} derive the moments of the {\em meta distribution}  of the downlink signal-to-interference ratio ($\sir$)  \cite{Martin2016Meta}.} However, \cite{salehi2018meta} ignores the joint decoding of the subset of layers associated with SIC.   
 {The authors of  \cite{ali2019downlinkLetter,salehi2018meta,salehi2018accuracy,tabassum2017modeling}  use the distribution of the typical link distance (in the network) to derive the order statistics of link distances of clustered users.  
As implied already, this ignores the correlation in the user locations placed in a PV cell which are a function of the BS point process.} A key unintended consequence of this approach is that it does not necessarily confine the user cluster in a PV cell, which is a significant approximation of the underlying setup  (see Fig. \ref{fig:illustration_PriorNOMAModel}, Middle and Left). 
The spectral efficiency of the $K$-tier heterogeneous cellular networks is analyzed in \cite{liu2017non} wherein the small cells serve their users using two-user NOMA with distance-based ranking. 
On similar lines, \cite{zhang2017downlink} derives the outage probability for  two-user downlink NOMA cellular networks, modeled as a PPP, by ranking the users based on the channel gains normalized by their received inter-cell  interference powers.  {The normalized gains are assumed to be independent and identically distributed (i.i.d.) and follow the distribution that is observed by a typical user in the network. This indeed ignores the correlation in the link distances as well as the inter-cell interference powers associated with the users placed in the same PV cell.}

A more reasonable way of accurately ranking the users is to form a user cluster by selecting users from distinct regions of the PV cell. One way of constructing these regions is based on the ratio of mean powers received from the serving and dominant interfering BSs. In particular, a PV cell can be divided into the cell center (CC) region, wherein the ratio is above a threshold $\tau$, and the cell edge (CE) region, wherein the ratio is below $\tau$. 
A similar approach of classifying users as CC and CE is also used in 3GPP studies to analyze the performance of schemes such as soft frequency reuse (SFR) \cite{dominique2010self}. Inspired by this, we characterize the CC and CE users based on their path-losses from the serving and dominant interfering BSs to pair them for the two-user NOMA system. While the proposed approach can be directly extended to the $N$-user NOMA using $N-1$ region partitioning thresholds, we focus on the two-user NOMA for the ease of exposition. {The proposed user pairing technique helps to construct the user cluster with distinct link qualities, which is essential for NOMA performance \cite{ding2016impact}. This is because of two key reasons: 1) the order statistic of received powers at different users in the cell is dominated by their  corresponding path-losses  \cite{wildemeersch2014successive},  and 2) the dominant interfering BS contributes most of the interference power  in the PPP setting \cite{Vishnu2017UAV}.}

Besides user ranking,  {resource allocation (RA)} is an integral part of the NOMA design.   On these lines, the prior works have investigated RA for the downlink NOMA system using a variety of performance metrics, such as  user fairness \cite{timotheou2015fairness,Choi2016Fairness,zhu2017optimal} and weighted sum-rate maximization \cite{zhu2017optimal,sun2017optimal,Parida_NOMA}.  {Further, the RA problems for maximizing the sum-rate \cite{zhu2017optimal,ali2019downlink,Wang2016PA} or the energy efficiency \cite{zhu2017optimal,Zhang2017EE,fang2016energy} subject to a minimum transmission rate constraint have also been investigated. }
 {Besides,  \cite{Chen_NOMA_OMA_RA} has demonstrated that the achievable {sum-rate} in  NOMA is always higher than that of the OMA under minimum transmission rate constraints.} {These RA formulations are meaningful for the full buffer  {non-real time (NTR)} services, such as file downloading, wherein transmission rates usually determine the QoS. However, they are not suitable for delay-sensitive applications with  {real time (RT)} traffic, such as video streaming and augmented reality, which are becoming even more critical in the context of newly emerging applications of wireless communication.} {For such applications, it is essential that the RA formulations explicitly include the delay constraints.} In this context, the effective capacity ($\ec$), defined in \cite{wu2003EffectiveCapacity} as the maximum achievable arrival rate that satisfies a delay QoS constraint, is analyzed for NOMA in \cite{yu2018link,Xiao2019LowLatency,Choi2017EffectiveCapacityDelayQoS} for the downlink case and in \cite{Qiao2012TransStrategy,Sebastian} for the uplink case. 
However, the existing works on the delay analysis of NOMA are relatively sparse. {Besides, the  works mentioned above on RA with the focus  on  $\csr$ (i.e., \cite{Chen_NOMA_OMA_RA,timotheou2015fairness,Choi2016Fairness,zhu2017optimal,Wang2016PA, Zhang2017EE,fang2016energy}) and $\ec$  (i.e., \cite{yu2018link,Xiao2019LowLatency,Choi2017EffectiveCapacityDelayQoS, Qiao2012TransStrategy,Sebastian})  are limited to the single-cell setting, and hence {they} ignore the impact of inter-cell interference.}
\subsection{Contributions}
\label{subsec:Contribution}
{The primary objective of this paper is to enable the accurate downlink NOMA analysis of cellular networks from the perspective of stochastic geometry.  {As we discussed in Section  \ref{subsec:PriorArt}, the existing stochastic geometry-based NOMA analyses do not capture correlation between signal qualities of the paired users that is induced by the fact that these users are located in the same PV cell. In addition, these analyses also do not explicitly pair users with distinct signal qualities which is crucial for harnessing performance gains using NOMA. The user pairing technique presented in this paper overcomes the above limitations while enabling tractable system-level analysis of downlink NOMA using stochastic geometry. }
Our approach focuses on the performance of the {\em typical cell}, which departs significantly from the standard approach of analyzing  the performance of the {\em typical user}   in a cellular network that is selected independently of the BS locations \cite{Praful_TypicalCell,Praful_TypicalCell_MetaDis}.  
The key contributions of our analysis are briefly summarized below. 
\begin{enumerate}
\item  {This paper presents a novel 3GPP-inspired user pairing technique for NOMA to accurately select the CC and CE users with distinct link qualities. }
\item We derive approximate yet accurate moments of the meta distributions for the CC and CE users belonging to the typical cell under both NOMA and OMA systems.  Besides, we also provide tight beta distribution approximations of the meta distributions.
\item Next, we  derive approximate distributions of the mean transmission rates and the upper bounds on the distributions of the mean packet delays of the CC and CE users for both NOMA and OMA systems under the random scheduling scheme.
\item Finally, we present RA formulations for NRT and RT services under both NOMA and OMA systems. For the NRT service, the objective is to maximize $\csr$ such that minimum transmission rates of the CC and CE users are satisfied. For the RT service, the aim is to maximize  {sum $\mathtt{EC}$ ($\sec$)} such that the minimum  $\ec$s of CC and CE services are achieved. We  present an efficient approach to obtain the near-optimal RAs for these formulations.
\item  Our numerical results demonstrate that: a) NOMA  is beneficial to both CC and CE users as compared to OMA except when the OMA schedules the CC users for most of the time, b) the proposed near-optimal RA under NRT services provide higher $\csr$ in NOMA as compared to that of the OMA, and c) the proposed near-optimal RA under RT services provide improved $\sec$ in NOMA as compared to that of the OMA at a higher user density.  
\end{enumerate} }
\section{System Model}
\label{sec:SystemModel}
\subsection{Network Modeling}
We model the locations of BSs  using the homogeneous PPP $\Phi$ with density $\lambda$.  {This paper considers the strongest mean power-based BS association policy. Hence, the coverage region of the BS at $\x\in\Phi$ becomes the PV cell $V_\x$ which is given by}  $V_\x=\{\y\in\mathbb{R}^2: \|\y-\x\|\leq \|\y-\x'\|,\forall \x'\in\Phi\}.$
Let $V_{\x c}$ and $V_{\x e}$ be the CC and CE regions, respectively, of the PV  cell $V_\x$ corresponding to the BS at $\x\in\Phi$, which are defined as 
\begin{equation}
\begin{split}
V_{\x c}&=\{\y\in V_\x : \|\y-\x\|\leq \min_{\x'\in\Phi_\x}\tau\|\y-\x'\|\}
\\\text{ and }
V_{\x e}&=\{\y\in V_\x : \|\y-\x\|> \min_{\x'\in\Phi_\x}\tau\|\y-\x'\|\},
\end{split}
\label{eq:CC_CE_Regions}
\end{equation}
where $\Phi_\x=\Phi\setminus\{\x\}$ and  $\tau\in(0,1)$ is the boundary threshold.  Fig. \ref{fig:illustration_PriorNOMAModel} (Left) depicts the CC and CE regions for $\tau=0.7$. Now, similar to \cite{Priyo2019FPR}, extending the application of the Type I user point process \cite{Haenggi2017}, we define the point processes of the locations of the CC and CE users  as 
\begin{equation}
\begin{split}
\Psi_{cc}&=\{U(V_{\x c};N_{\x c}): \x\in \Phi \}\\
\text{and}~\Psi_{ce}&=\{U(V_{\x e};N_{\x e}): \x\in \Phi \},
\end{split}
\label{eq:CC_CE_pps}
\end{equation}
respectively, where $U(A;N)$ denotes $N$ points chosen independently and uniformly at random from the set $A$.  {Here,  $N_{\x c}$ and $N_{\x e}$ are the numbers of CC users in $V_{oc}$ and CE users in $V_{oe}$, respectively, where $\mu$ represents the user density. We assume that  $N_{\x c}$ and $N_{\x e}$ follow the zero-truncated Poisson distributions with means $\nu|V_{\x c}|$ and $\nu|V_{\x e}|$, respectively.}   We refer to $N_{\x c}$ and $N_{\x e}$ as the {\em CC} and {\em CE loads} of the BS at $\x$.   {Table I summarizes the system design variables and performance metrics considered in this paper. }

Since by the Slivnyak\textquotesingle s theorem, conditioning on a point is the same as adding a point to the PPP, we consider that the nucleus of the {\em typical cell} of the point process $\Phi\cup\{o\}$ is located at the origin $o$.  Thus, the typical cell becomes $
V_o=\{\y\in\mathbb{R}^2\mid \|\y-\x\|>\|\y\|\,\forall \x\in\Phi\}$.
For this setting, the locations the {\em typical CC user} and the {\em typical CE user} of $\Psi_{cc}$ and $\Psi_{ce}$ can be modeled using the uniformly distributed points in $V_{oc}$ and $V_{oe}$, respectively. Thus, $\Phi$ becomes the point process of the interfering BSs that is observed by the typical CC and CE users.

 \begin{table*}
\centering
\caption{}
\label{table:Syatem_Variable}
\hspace{-5mm}{\small \begin{tabular}{ |c |c||c|c| }
\hline
\multicolumn{2}{|c||}{\textbf{System design variables }} & \multicolumn{2}{c|}{\textbf{Performance metrics}}\\ \hline 
 BS point process  & $\Phi$  & Meta distribution for CC users & $\bar{F}_c(\cdot,\cdot)$ \\ \hline
 BS and user densities & $\lambda$ and $\nu$ & Meta distribution for CE users  &  $\bar{F}_e(\cdot,\cdot)$ \\ \hline
 CC and CE users point processes   & $\Psi_{cc}$ and  $\Psi_{ce}$ & $b$-th moments of meta distributions & $M_b^c$ and $M_b^e$\\  \hline
 PV cell associated with BS at $\mathbf{x}$ & $V_\mathbf{x}$ & Transmission rates of CC and CE users & $R_c$ and $R_e$\\ \hline
  CC and CE regions of PV cell $V_x$   & $V_{\mathbf{x}c}$ and  $V_{\mathbf{x}e}$ & Service rates of CC and CE users & $\mu_c$ and $\mu_e$ \\  \hline
 Number of CC and CE users     & $N_{\mathbf{x}c}$ and $N_{\mathbf{x}e}$ & $\mathtt{CDF}$ of CC user's transmission rate & $\mathcal{R}_c(\cdot,\cdot)$ \\  \hline
 Service link distance   & $R_o$  &  $\mathtt{CDF}$ of CE user's transmission rate & $\mathcal{R}_e(\cdot,\cdot)$\\ \hline
 Dominant interfering link distance   & $R_d$  & $\mathtt{CDF}$ of upper bound on CC user's delay & $\mathcal{D}_c(\cdot,\cdot)$  \\ \hline
Boundary threshold for regions  & $\tau$  & $\mathtt{CDF}$ of upper bound on CE user's delay &   $\mathcal{D}_e(\cdot,\cdot)$  \\ \hline
 Transmission layers  & $\mathtt{L_c}$ and $\mathtt{L_e}$  & Cell sum rate      & $\mathtt{CSR_{NOMA}}$\\ \hline
 $\mathtt{SIR}$ thresholds & $\beta_c$ and $\beta_e$ &  Effective capacity of CC service  & $\mathtt{EC}^c_{\mathtt{NOMA}}$  \\ \hline
 Power allocated to $\mathtt{L_c}$ and $\mathtt{L_e}$ & $\theta $ and $(1-\theta)$ &  Effective capacity of  CE service  &  $\mathtt{EC}^e_{\mathtt{NOMA}}$ \\ \hline
\end{tabular}}
\end{table*} 

Let $R_o=\|\y\|$ be the {\em service link distance}, i.e., the distance between the user at $\y\in V_o$ and its serving BS at $o$. Let $R_d=\|\x_d-\y\|$ be the distance from the user at $\y\in V_o$ to its dominant interfering BS at $\x_d\in\Phi$ where $\x_d=\argmax_{\x\in\Phi}\|\x-\y\|^{-\alpha}$ and $\alpha$ is the path-loss exponent. Therefore, the definitions, given in \eqref{eq:CC_CE_Regions} and \eqref{eq:CC_CE_pps}, implicitly classify the CC and CE users  based on their distances (i.e., path-losses) from their serving and dominant interfering BSs such that the CC user at $\y\in V_{oc}$ has $R_o\leq \tau R_d$ and the CE user at $\y\in V_{oe}$ has $R_o> \tau R_d$.
\begin{figure*}
 \centering  
\hspace{-1.5cm}  \includegraphics[trim=.85cm 1cm .85cm .5cm, width=.35\textwidth]{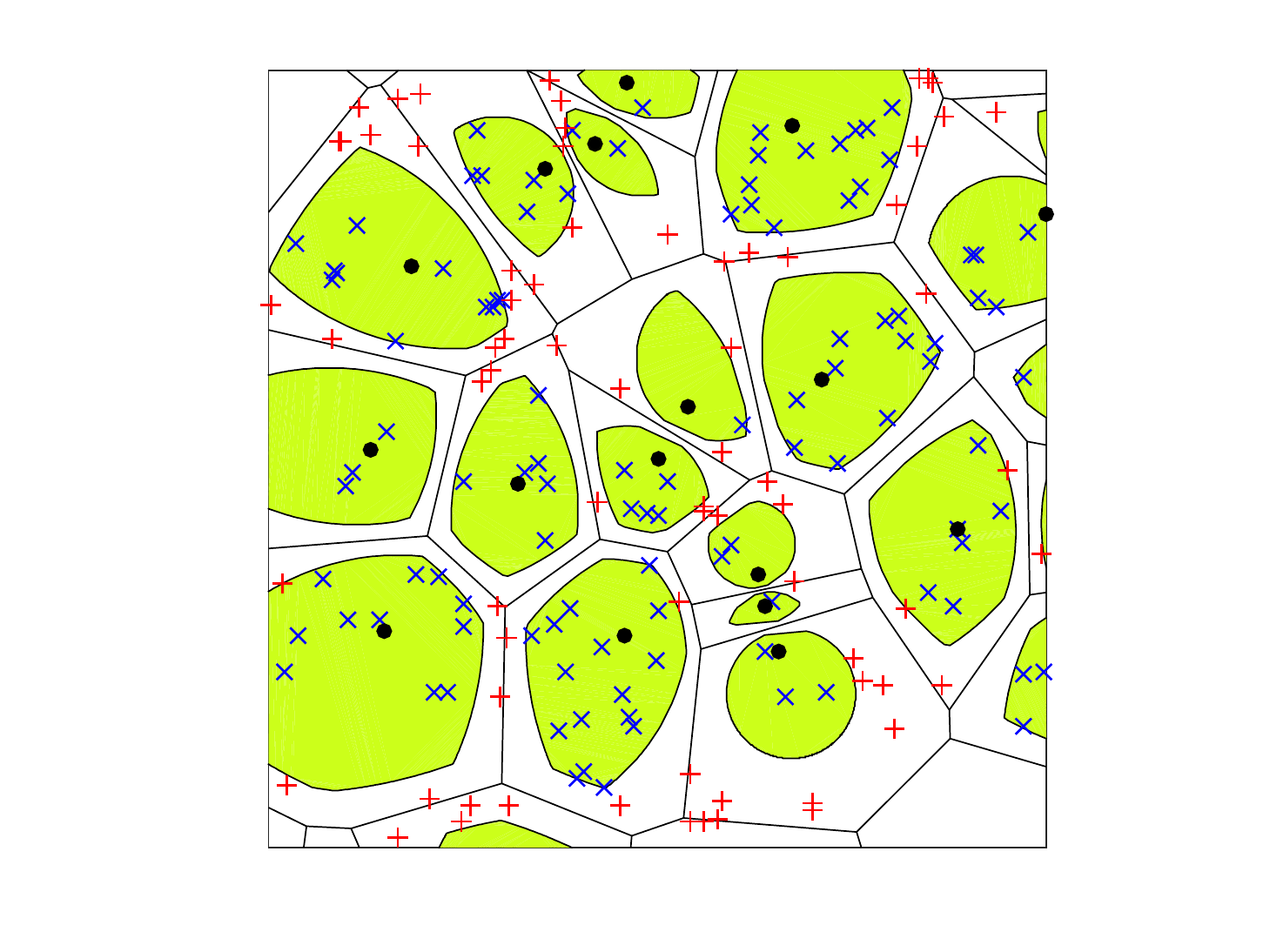}\hspace{-.5cm} 
\includegraphics[trim=.85cm 1cm .85cm .5cm, width=.35\textwidth]{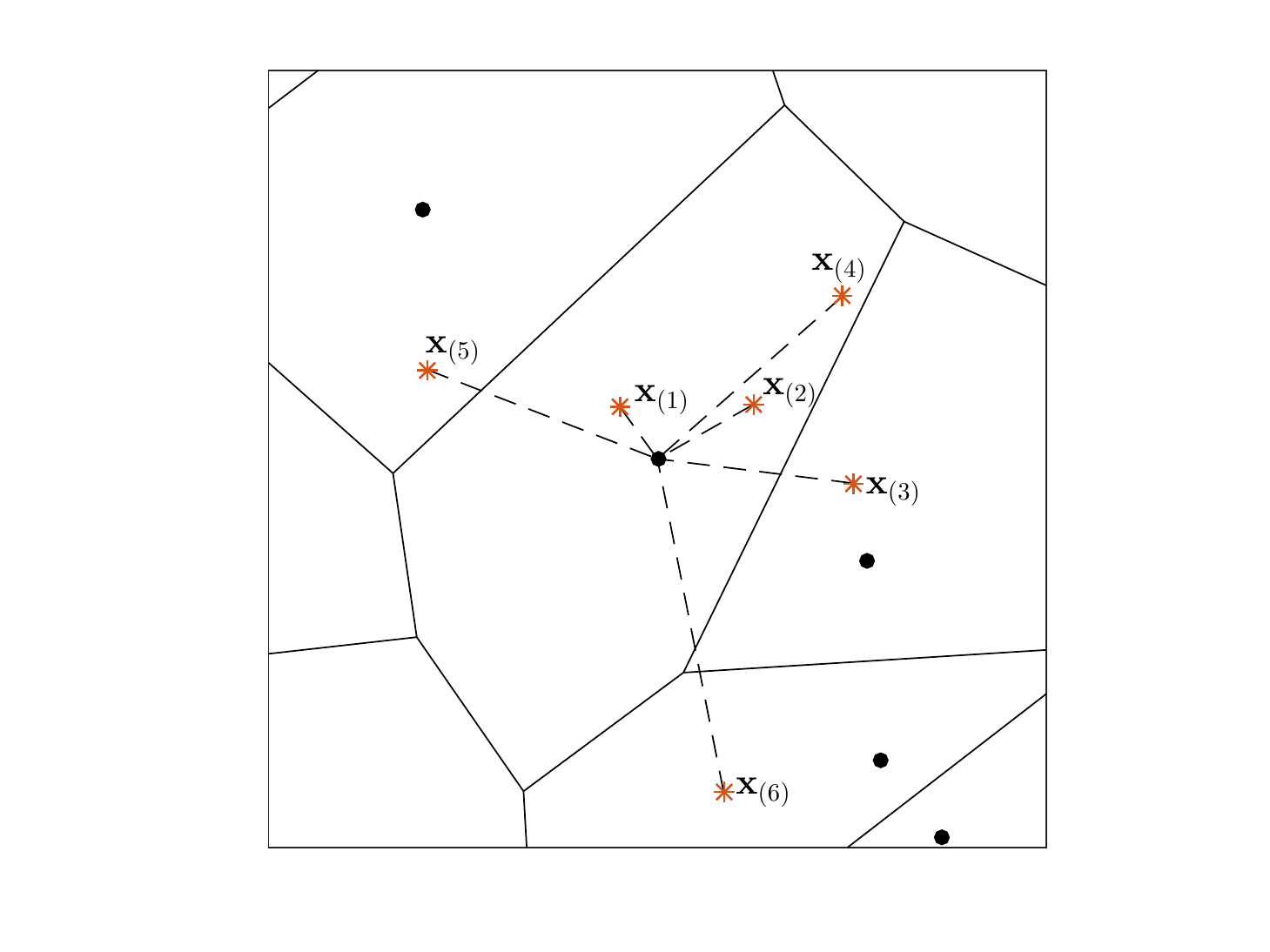}\hspace{-.1cm} 
\includegraphics[trim=.85cm 1cm .85cm .5cm, width=.35\textwidth]{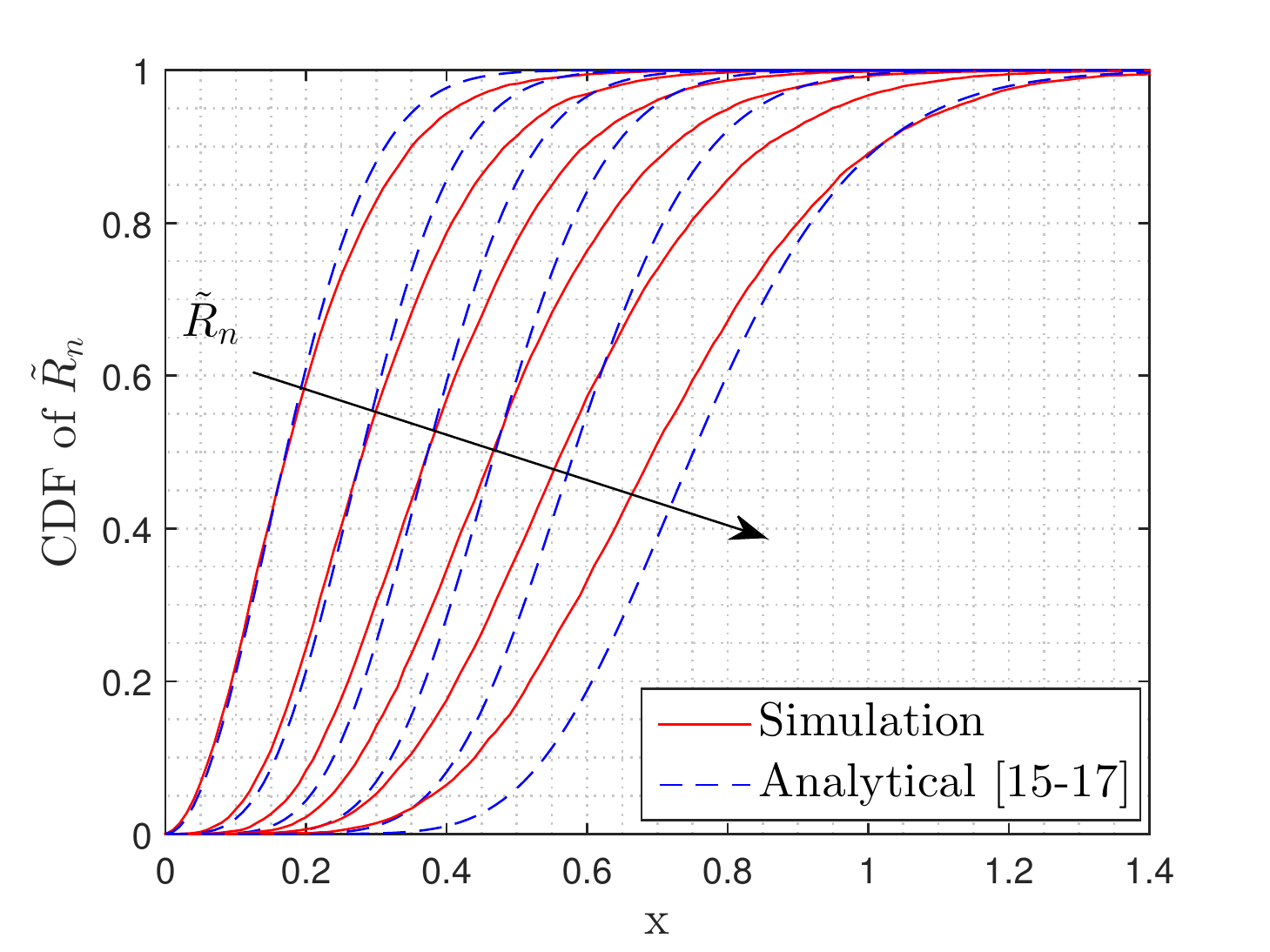}\vspace{-.2cm}
\caption{Left: typical realization of $\Psi_{cc}$ and $\Psi_{ce}$ for $\tau=0.7$, $\lambda=1$, and $\nu=20$.  Middle:  an illustration of the user cluster from  \cite{ali2019downlinkLetter,salehi2018meta,salehi2018accuracy} for $N=6$ (and the fact that the PV cell not necessarily confines the user cluster). Right: the distributions of the ordered link distances modeled in \cite{ali2019downlinkLetter,salehi2018meta,salehi2018accuracy} for $N=6$ and $\lambda=1$. The dot, cross, plus, and star markers correspond to BSs, CC users, CE users, and user cluster,  respectively. The green and white colors in the left figure show the CC and CE regions.}\vspace{-3mm}
\label{fig:illustration_PriorNOMAModel}
\end{figure*}
Fig. \ref{fig:illustration_PriorNOMAModel} (Left) illustrates a typical realization of $\Psi_{cc}$ and $\Psi_{ce}$. From this figure, it is clear that \eqref{eq:CC_CE_Regions} accurately preserves the
CC and CE regions wherein the $\sir$ is expected to be higher and lower, respectively. As a comparison, Fig. \ref{fig:illustration_PriorNOMAModel} (Middle) illustrates a realization of a user cluster that results from the distance-based ranking scheme {that is employed in}  \cite{ali2019downlinkLetter,salehi2018meta,salehi2018accuracy}. As is clearly evident from the figure, the user cluster is not confined to the PV cell, which is an unintended consequence of ignoring correlation in the user locations. 
This can also be verified by comparing the distributions of the ordered distances used in \cite{ali2019downlinkLetter,salehi2018meta,salehi2018accuracy} with those obtained from the simulations. This comparison is given in Fig. \ref{fig:illustration_PriorNOMAModel} (Right) wherein $\tilde{R}_n$ is the link distance of the $n$-th closest user from the BS. 

It should be noted here that {the user clustering described in the above is different} from the {\em geographical} clustering of users and BSs that have been recently modeled using PCPs, e.g., see \cite{Mehrnaz2018,Praful_UserClustering}. While the PCPs are useful in capturing the spatial coupling in the locations of the users and BSs, they do not necessarily partition users in different groups based on their QoS as done in the above scheme. Therefore, the proposed technique to form user clusters  can apply to more general studies that may require to partition the users based on their QoS experiences, such as soft frequency reuse \cite{dominique2010self}. Now, we discuss the downlink NOMA transmission for a randomly selected pair of the CC and CE users in the typical cell in the following subsection.
\subsection{Downlink NOMA Transmission based on the Proposed User Pairing}
Each BS is assumed to transmit signal superimposed of two layers corresponding to the messages for the CC and CE users forming a user cluster. Henceforth, the layers intended for the CC and CE users are referred to as the $\mathtt{L}_c$ and $\mathtt{L}_e$ layers, respectively.
The $\mathtt{L}_c$ and $\mathtt{L}_e$ layers are encoded at power levels of $\theta P$ and $(1-\theta)P$, respectively, where $P$ is the transmission power per RB and $\theta \in(0,1)$. Without loss of generality, we assume $P=1$ (since we ignore thermal noise). Usually, NOMA allocates more power to the weaker user so that it receives smaller intra-cell interference power compared to the desired signal power. 
Thus, the CC user first decodes the $\mathtt{L}_e$ layer while treating the power assigned to the $\mathtt{L}_c$ layer as interference. After successfully decoding the $\mathtt{L}_e$ layer, the CC user cancels its signal using SIC from the received signal and then decodes the $\mathtt{L}_c$ layer. 
The $\sir$s of the typical CC user at $\y\in V_{oc}$ for decoding the $\mathtt{L}_c$ and $\mathtt{L}_e$ layers are
\begin{align}
\sir_{e}&=\frac{h_{o}R_o^{-\alpha}(1-\theta)}{\theta h_{o}R_o^{-\alpha}+I_{\Phi}}~
\text{and~} \sir_{c}=\frac{h_{o}R_o^{-\alpha}\theta}{I_{\Phi}},
\label{eq:SIR}
\end{align}
respectively, where $I_{\Phi}=\sum_{\x\in\Phi}  h_{\x}\|\x-\y\|^{-\alpha} $ represents the aggregate inter-cell interference resulting from the set of BSs in $\Phi$ and $h_{\x}\sim \exp(1)$ are i.i.d. fading channel gains. 
Besides, the CE user decodes {the} $\mathtt{L}_e$ layer while treating the power assigned to the $\mathtt{L}_c$ layer as  interference. Thus, the effective $\sir$ of the typical CE user at $\y\in V_{oe}$ is also $\sir_{e}$ {as} given in \eqref{eq:SIR}. Note that we assumed a full-buffer system in \eqref{eq:SIR}, which is a common assumption in the stochastic geometry literature and is quite reasonable for NRT services for which the objective is to maximize the {sum-rate}. We will further argue in Subsection \ref{subsec:SchedulinThroughputDelay} that this simple setting also provides useful bounds on the delay-centric metrics. 
\subsection{Meta Distribution for the Downlink NOMA System}
\label{subsec:MetaDistribution}
The success probabilities for the CC and CE users are defined as the probabilities that the typical CC and CE users can decode their intended messages.  {These success probabilities provide the mean performance of the typical CC and CE users in the network. However, they do not give any information on the disparity in the link performance of the CC and CE users spread across the network.}  For that purpose,  the distribution of the conditional success probability (conditioned on the locations of BS  w.r.t. to the CC/CE user location) can be useful. The distribution of the conditional success probability is referred to as the {\em meta distribution} \cite{Martin2016Meta}. The meta distribution for the CC/CE user can be used to answer questions like ``what percentage of the CC/CE users can establish their links with the transmission reliability above a predefined threshold for a given $\sir$ threshold?". Building on the  definition of the meta distribution in \cite{Martin2016Meta}, we define the meta distributions for the CC and CE users  under NOMA as below.
\begin{definition}
The meta distribution of the typical CC user's success probability is defined as 
\begin{equation}
\bar{F}_{\text{c}}(\beta_c,\beta_e;x)=\P[\p_c(\beta_c,\beta_e\mid\y, \Phi)>x],
\end{equation}
and the meta distribution of the typical CE user's success probability is defined as 
\begin{equation}
\bar{F}_{\text{e}}(\beta_e;x)=\P[\p_e(\beta_e\mid\y,\Phi)>x],
\end{equation}
where $x\in[0,1]$, $\beta_c$ and $\beta_e$ are the $\sir$ thresholds corresponding to the $\mathtt{L}_c$ and $\mathtt{L}_e$ layers, respectively.
$p_c(\beta_c,\beta_e\mid\y,\Phi)=\P[\sir_{c}\geq \beta_c,\sir_{e}\geq \beta_e\mid\y, \Phi]$ and $p_e(\beta_e\mid \y,\Phi)=\P[\sir_{e}\geq \beta_e\mid \y,\Phi]$ are the success probabilities of the typical CC and CE users conditioned on their locations at $\y$ and the point process $\Phi$ of the interfering BSs, respectively. 
\end{definition}

\subsection{Traffic Modeling, Scheduling, and  Performance Metrics}
\label{subsec:SchedulinThroughputDelay}
   For the above setting, we will perform  a comprehensive load-aware performance analysis of the CC/CE users under the NOMA system. For this,  we consider {\em random scheduling} wherein each BS randomly selects a pair of CC and CE users within its PV cell for the NOMA transmission in a given time slot. For the NRT services, our objective is to maximize the sum-rate.  Since the network is assumed to be static,  the load-aware transmission rate of CC/CE user at $\y$ depends on its {\em scheduling probability} and {\em successful transmission probability}, both conditioned on $\Phi$. {As already implied above, each CC/CE user within a given cell is equally likely to be scheduled in each time slot. Besides, $\sir$s experienced by the CC/CE user at $\y$ are i.i.d.~across the time slots for given $\Phi$.}  
  Each BS transmits signal superimposed with $\mathtt{L}_c$ and $\mathtt{L}_e$ layers which are encoded at rates $\mathtt{B}\log_2(1+\beta_c)$ and $\mathtt{B}\log_2(1+\beta_e)$ bits/sec, respectively, where $\mathtt{B}$ is the channel bandwidth. Without loss of generality, here onwards we assume $\mathtt{B}=1$.
 Therefore,  according to Shannon's capacity law, the user requires $\sir_c$ and $\sir_e$  above thresholds $\beta_c$ and $\beta_e$, respectively, to successfully decode these layers.  Hence, for given $\Phi$, the achievable transmission rates of the CC and CE users located at $\y\in V_o$  under random scheduling respectively become 
 \begin{equation}
 \begin{split}
R_c(\y,\Phi)&=\frac{p_c(\beta_c,\beta_e\mid\y,\Phi)}{N_{oc}} \log_2(1+\beta_c)\\ \text{and~}R_e(\y,\Phi)&=\frac{p_e(\beta_e\mid\y,\Phi)}{N_{oe}} \log_2(1+\beta_e).
\end{split}\label{eq:RateCondPhi_CC_CE}
\end{equation}
\indent The characterization of the typical CC/CE user's success probability under the full-buffer system is also useful in obtaining an upper bound on its  delay performance for more general traffic patterns. This can subsequently be used to derive lower bounds on the $\sec$ for the given RT service.  {The exact delay analysis for the cellular networks is known to be challenging because of the coupled queues at different BSs. Generally, the performance of coupled queues is studied by employing meaningful {\em modifications} to the system \cite{Rao_1988}; see \cite{bonald2004wireless,Haenggi_Stability,Zhong_SpatoiTemporal} for a small subset of relevant works in the context of cellular networks.} {For a modified system, the general approach includes the full buffer assumption for the interfering links. As a result, the queue associated with the link of interest operates independently of the statues of queues associated with the interfering links. Note that such a modified system is consistent with our system set-up that is assumed from the very beginning.}  Since this effectively underestimates the success probability, it provides an upper bound on the packet transmission delay. 
In the same spirit, this paper presents an upper bound on the distribution of the conditional mean delay of the typical link, which is  tighter  for the higher load scenarios.  {One can, of course, determine the mean delay of the typical link by adopting more sophisticated analyses, such as the mean cell approach \cite{Bartllomiej2016}.  However, the contributions of this paper revolve around  accurate downlink NOMA analysis using a new user pairing scheme because of which these other approaches are out of the scope of the current paper. }
{We assume that the typical CC/CE user has a dedicated queue of infinite length which is placed at its serving BS. The packet arrival process of the CC/CE  service  is assumed to follow the Bernoulli distribution with {a} mean of $\varrho_c$/$\varrho_e$ packets per slot. The  packet sizes of the CC and CE services are considered  to be equal to $\mathtt{TB}\log_2(1+\beta_c)$ and $\mathtt{TB}\log(1+\beta_e)$ bits, respectively, where $\mathtt{T}$ is the slot duration. } Note that the generalized values of $\sir$ thresholds $\beta_c$ and $\beta_e$ facilitate the choice of selecting CC and CE services with different packet sizes.   
{Thus, the successful packet transmission rates (in packets per slot) for the typical CC and CE users at $\y$ given $\Phi$  respectively become }
\begin{equation}
\begin{split}
\mu_{c}(\y,\Phi) &=\frac{p_c(\beta_c,\beta_e\mid\y,\Phi)}{N_{oc}} \\
\text{ and~} \mu_{e}(\y,\Phi) &= \frac{p_e(\beta_e\mid\y,\Phi)}{N_{oe}}.
\end{split}\label{eq:DelayCondPhi_CC_CE}
\end{equation} 
\indent Besides, we also analyze the load-aware performances of the CC and CE users under above discussed services for the OMA system.  {For OMA system, we consider that each BS  schedules  one of its associated CC (CE) users that is chosen uniformly at random in a given time slot if $r\leq \eta$ ($r>\eta$) where $r\sim U([0,1])$ is generated independently across the time slots.}
The parameter $\eta$ allows to control the frequency of scheduling of the CC and CE users in order to meet their QoS requirements. The CC/CE load of the typical cell, and hence the scheduling probability of typical CC/CE user, depends on $\Phi$ (see \eqref{eq:CC_CE_pps}). 
Thus, from \eqref{eq:RateCondPhi_CC_CE}-\eqref{eq:DelayCondPhi_CC_CE}, it is quite evident that the exact analysis requires the joint statistical characterization of the success probability and scheduling probability. However, such joint characterization is challenging as the distribution (even the moments)  of the area of the PV cell $V_o$ conditioned on  $\y\in V_o$ is difficult to obtain. Hence, similar to \cite{Zhong_SpatoiTemporal,Priyo2019FPR}, we adopt the following reasonable assumption in our analysis.
\begin{assumption}
\label{assumption:Independence_PVCellArea_SuccessProb}
We assume that the  CC/CE load (or,  the scheduling probability) and the successful transmission probability observed by the typical CC/CE user are independent.  
\end{assumption}
The numerical results presented in Section \ref{sec:NumericalResults} will demonstrate the accuracy of this assumption for the analysis of the metrics discussed above. 
\section{Meta Distribution Analysis for the CC and CE users}
\label{sec:MetaDistributionsAnalysis}
The main goal of this Section is to present the downlink meta distribution analysis for the NOMA and OMA systems. As stated already in Section \ref{subsec:Contribution}, we characterize the performance of the {\em typical cell} which departs significantly from the standard stochastic geometry approach of analyzing the performance of the typical user. The key intermediate step in the meta distribution analysis is the joint characterization of the service link distance $R_o=\|\y\|$, where $\y\sim U(V_o)$, and the point process $\Phi$ of the interfering BSs. Hence, to enable the analysis of the meta distributions of the CC and CE users in the typical cell, we require the joint characterizations of $R_o$ and $\Phi$ under the conditions of $R_o\leq R_d\tau$ and $R_o>R_d\tau$. For this, we first determine the marginal and joint probability density functions ($\pdf$s) of $R_o$ and $R_d$ for the CC  and CE users.
However, given the complexity of the analysis of r.v. $R_o$ \cite{PraPriHar}, it is reasonable to assume that the exact joint characterization of $R_o$ and $R_d$ is equally, if not more, challenging.
The marginal distribution of $R_o$ is generally approximated using the contact distribution {with  adjusted density by} a correction factor (c.f.) to maintain tractability  \cite{Haenggi2017}. Thus, we approximate the joint $\pdf$ of $R_o$ and $R_d$ using the joint $\pdf$ of the distances to the two nearest points in PPP as 
\begin{align}
f_{R_o,R_d}(r_o,r_d)&=(2\pi\rho\lambda)^2r_or_d\exp(-\pi\rho\lambda r_d^2),
\label{eq:pdf_RoRd}
\end{align}
for  $r_d\geq r_o\geq 0$, where $\rho=\frac{9}{7}$ is the c.f.~(refer \cite{Praful_TypicalCell} for more details).
\begin{lemma}
\label{lemma:DistanceDistributions}
The probabilities that a user uniformly distributed in the typical cell is the CC user and the CE user are equal to $\tau^2$ and $1-\tau^2$, respectively. The cumulative density function ($\cdf$) of $R_o$ and the $\cdf$ of $R_d$ conditioned on $R_o$ for the CC user are given by 
\begin{equation}
F^{\text{c}}_{R_o}(r_o)=1-\exp\left(-\pi\rho\lambda {r_o^2}/{\tau^2}\right),\label{eq:CDFRo_CC1}
\end{equation}
for $r_o>0$, and
\begin{equation}
F^{\text{c}}_{R_d\mid R_o}\left(r_d\mid r_o\right)=1-\exp\left(-\pi\rho\lambda \left(r_d^2-{r_o^2}/{\tau^2}\right)\right),\label{eq:CDFRo_CC2}
\end{equation}
for $r_d>\frac{r_o}{\tau}$, respectively. The joint $\pdf$ of $R_o$ and $R_d$ for the CC user is given by
\begin{equation}
f^{\text{c}}_{R_o,R_d}(r_o,r_d)=\frac{(2\pi\rho\lambda)^2}{\tau^2}r_or_d\exp\left(-\pi\rho\lambda r_d^2\right),
\label{eq:pdf_RoRd_CC}
\end{equation}
for $r_d>\frac{r_o}{\tau}$ and $r_o>0$.
The $\cdf$ of $R_o$ and the $\cdf$ of $R_d$ conditioned on $R_o$ for the CE user are respectively given by 
\begin{equation}
F^{\text{e}}_{R_o}(r_o)=1-\frac{1-\tau^2\exp\left(-\pi\rho\lambda{r_o^2}(\tau^{-2}-1)\right)}{(1-\tau^2)\exp\left(\pi\rho\lambda{r_o^2}\right)},\label{eq:CDFRo_CE}
\end{equation}
for $r_o> 0$, and $F^{\text{e}}_{R_d\mid R_o}\left(r_d\mid r_o\right)$
\begin{align}
=\begin{cases}\frac{1-\exp(-\pi\rho\lambda(r_d^2-r_o^2))}{1-\exp(-\pi\rho\lambda r_o^2(\tau^{-2}-1))},~ &\text{for}~\frac{r_o}{\tau}> r_d\geq r_o,\\
1, &\text{for}~r_d\geq \frac{r_o}{\tau}.
\end{cases}
\label{eq:CDFRdCondRo_CE}
\end{align}
The joint $\pdf$ of $R_o$ and $R_d$ for the CE user is given by
\begin{equation}
f^{\text{e}}_{R_o,R_d}(r_o,r_d)=\frac{(2\pi\rho\lambda)^2}{1-\tau^2}r_or_d\exp\left(-\pi\rho\lambda r_d^2\right),
\label{eq:pdf_RoRd_CE}
\end{equation}
for $\frac{r_o}{\tau}\geq r_d>r_o$ and $r_o>0$.
\end{lemma}
\begin{proof}
Please refer to  Appendix A.
\end{proof}
\subsection{Meta Distribution for CC and CE users}
\label{sec:MetaDistributions_NOMA_OMA}
{The CC user needs to decode both the $\mathtt{L}_c$ and $\mathtt{L}_e$ layers for the successful reception its own message. Thus, the successful transmission event for the CC user becomes}
\begin{align}
\ncalE_c&=\{\sir_{c}>\beta_c\}\cap\{\sir_{e}>\beta_e\}\nonumber\\
&=\left\{h_o>R_o^\alpha I_\Phi\chi_c\right\},
\label{eq:SuccessEevent_CC}
\end{align}
where $\chi_c=\max\left\{\frac{\beta_c}{\theta},\frac{\beta_e}{1-\theta(1+\beta_e)}\right\}$. On the other hand, the CE user decodes its message while treating the signal intended for the CC user as the interference power. Thus, the successful transmission event for the CE user is given by
\begin{align}
\ncalE_e&=\{\sir_{e}>\beta_e\}=\left\{h_{o}>R_o^\alpha I_{\Phi}\chi_e\right\},
\label{eq:SuccessEevent_CE}
\end{align}
where $\chi_e=\frac{\beta_e}{1-\theta(1+\beta_e)}$. 
Since it is difficult to derive the meta distribution \cite{Martin2016Meta}, we first derive its moments in the following theorem and then use them to approximate the meta distribution.
\begin{theorem}
\label{thm:MetaDisMoment}
The $b$-th moments of the meta distributions of conditional success probability for the typical CC and CE users under NOMA respectively are
 
 \begin{equation}
\begin{split}
M_b^{\text{c}}(\chi_c)&=\frac{\rho^2}{\tau^2}\int\limits_0^{\tau^2} \frac{\left(\rho + v\ncalZ_b\left(\chi_c,v\right)\right)^{-2}}{(1+\chi_cv^{\frac{1}{\delta}})^b}{\rm d}v,\\
\text{and}~M_b^{\text{e}}(\chi_e)&=\frac{\rho^2}{1-\tau^2}\int\limits_{\tau^2}^1 \frac{\left(\rho + v\ncalZ_b\left(\chi_e,v\right)\right)^{-2}}{(1+\chi_e v^\frac{1}{\delta})^b}{\rm d}v,\label{eq:MetaDisMoment_CC_CE}
\end{split}
\end{equation}
where~ $\ncalZ_b(\chi,a)=\chi^\delta\int_{\chi^{-\delta} a^{-1}}^\infty \left[1-({1+t^{-\frac{1}{\delta}}})^{-b}\right]{\rm d}t$ and \\$\delta=\frac{2}{\alpha}$.
\end{theorem}
\begin{proof}
Please refer to Appendix B.
\end{proof}
In OMA, each BS serves its associated users using orthogonal RBs which means that there is no intra-cell interference.  Thus, OMA provides better success probabilities for the CC and CE users compared to NOMA. However, the orthogonal RB allocation reduces the transmission instances for the CC and CE users, which in turn affects their transmission rates negatively.  The successful transmission events for the CC and CE user under OMA respectively  given by
\begin{align}
\tilde{\ncalE}_c&=\{h_o>R_o^\alpha \beta_c I_\Phi\} \text{ and } \tilde{\ncalE}_e=\{h_o>R_o^\alpha \beta_e I_\Phi\}.
\label{eq:SuccessEevent_CC_CE_OMA}
\end{align}
The following corollary presents the  $b$-th moments of the meta distributions for the OMA case. 
\begin{cor}
\label{cor:OrthogonalTransmission}
The $b$-th moments of the meta distributions of conditional success probability for the  typical CC and CE users under OMA respectively are
 \begin{align}
 \begin{split}
\tilde{M}_b^{\text{c}}(\beta_c)&=\frac{\rho^2}{\tau^2}\int\limits_0^{\tau^2} \frac{\left(\rho + v\ncalZ_b\left(\beta_c,v\right)\right)^{-2}}{(1+\beta_cv^{\frac{1}{\delta}})^b}{\rm d}v,\\
\text{and}~\tilde{M}_b^{\text{e}}(\beta_e)&=\frac{\rho^2}{1-\tau^2}\int\limits_{\tau^2}^1 \frac{\left(\rho + v\ncalZ_b\left(\beta_e,v\right)\right)^{-2}}{(1+\chi_e v^\frac{1}{\delta})^b}{\rm d}v,
\end{split}\label{eq:MetaDisMoment_CC_CE_OMA}
\end{align}
where~ $\ncalZ_b(\beta,a)=\beta^\delta\int_{\beta^{-\delta} a^{-1}}^\infty \left[1-({1+t^{-\frac{1}{\delta}}})^{-b}\right]{\rm d}t$.
\end{cor} 
\begin{proof}
Using the definitions in \eqref{eq:SuccessEevent_CC_CE_OMA} and following the steps in Appendix B, we obtain \eqref{eq:MetaDisMoment_CC_CE_OMA}. 
\end{proof}
Fig. \ref{fig:Moments_NOMA_OMA_BetaApp} verifies that the means and variances of the meta distributions of the CC and CE users under NOMA (Left) and OMA (Middle) derived in Theorem \ref{thm:MetaDisMoment} and Corollary \ref{cor:OrthogonalTransmission}, respectively, closely match with the simulation results. 
The moments  for the CE user monotonically decrease with $\theta$ 
as the interference from the $\mathtt{L}_c$ layer increases with $\theta$.
However, the performance trend of the moments for the CC user  w.r.t. $\theta$ is different. {This is because $\theta$ affects the probabilities of decoding the $\mathtt{L}_c$ and $\mathtt{L}_e$ layers  at the CC user differently.} While increasing $\theta$ makes it difficult to decode $\mathtt{L}_e$ layer, it makes it easier to decode $\mathtt{L}_c$ layer. As a result, the impact of $\mathtt{L}_c$ layer decoding is dominant for $\theta\leq\hat\theta~(=0.5$ for $(\beta_c,\beta_e)=(0,-3)$ (in dB)) and the impact of $\mathtt{L}_e$ layer decoding is dominant for $\theta>\hat\theta$ where $\hat\theta$ will be defined in Section \ref{sec:CSEoptimization}.
\begin{figure*}
 \centering
\hspace{-.35cm} \includegraphics[width=.33\textwidth]{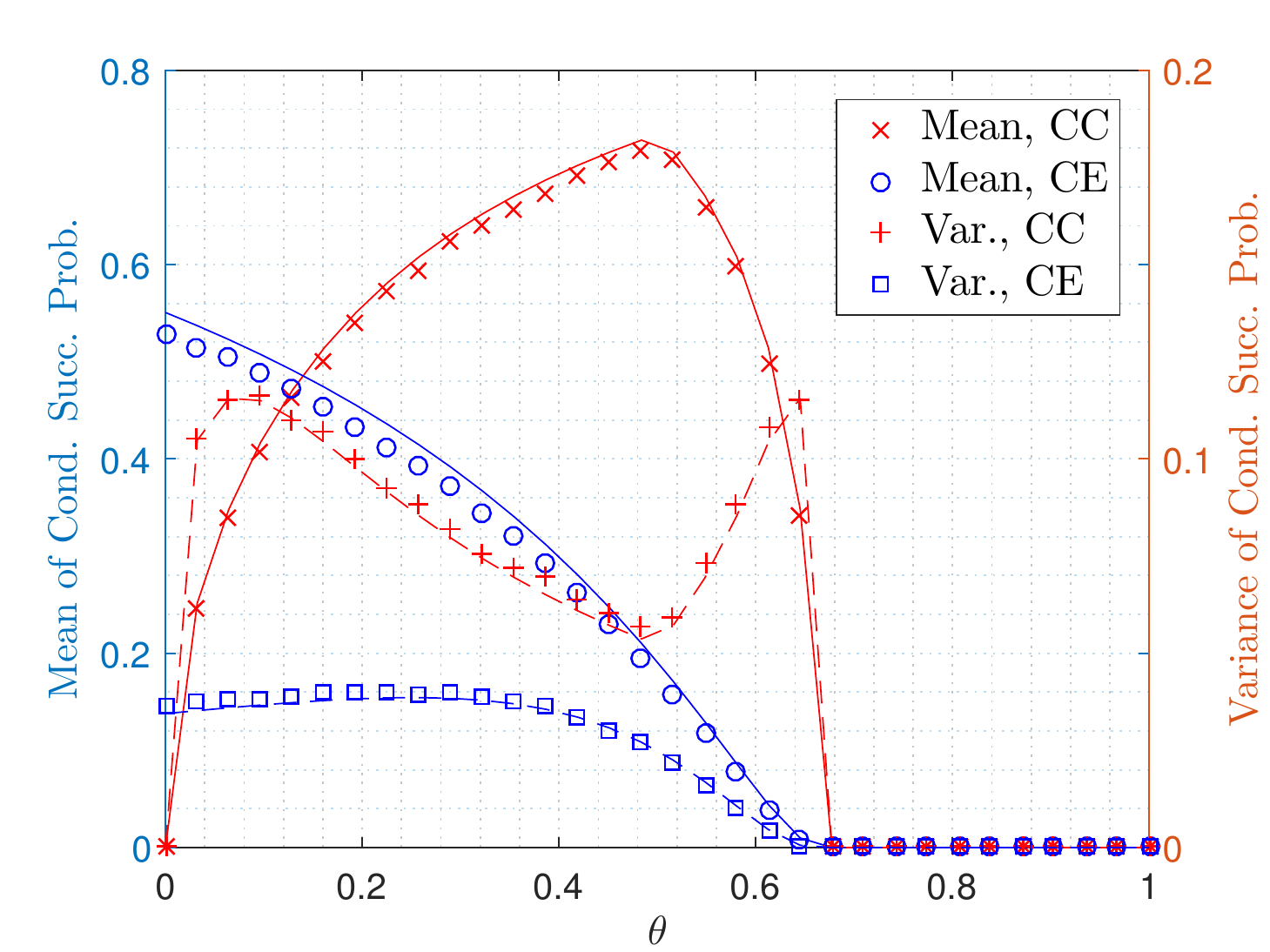}
\includegraphics[width=.33\textwidth]{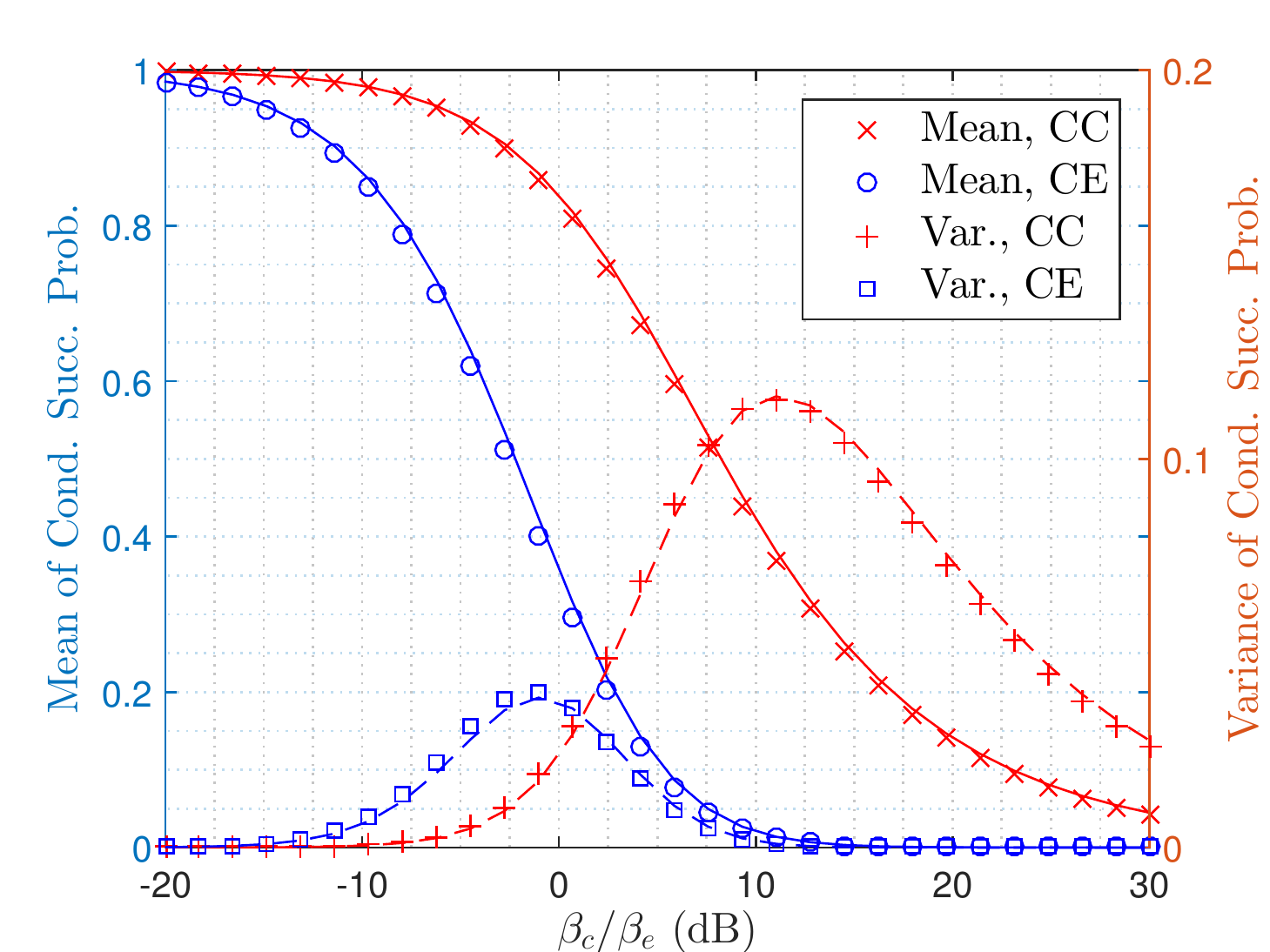}
\includegraphics[width=.33\textwidth]{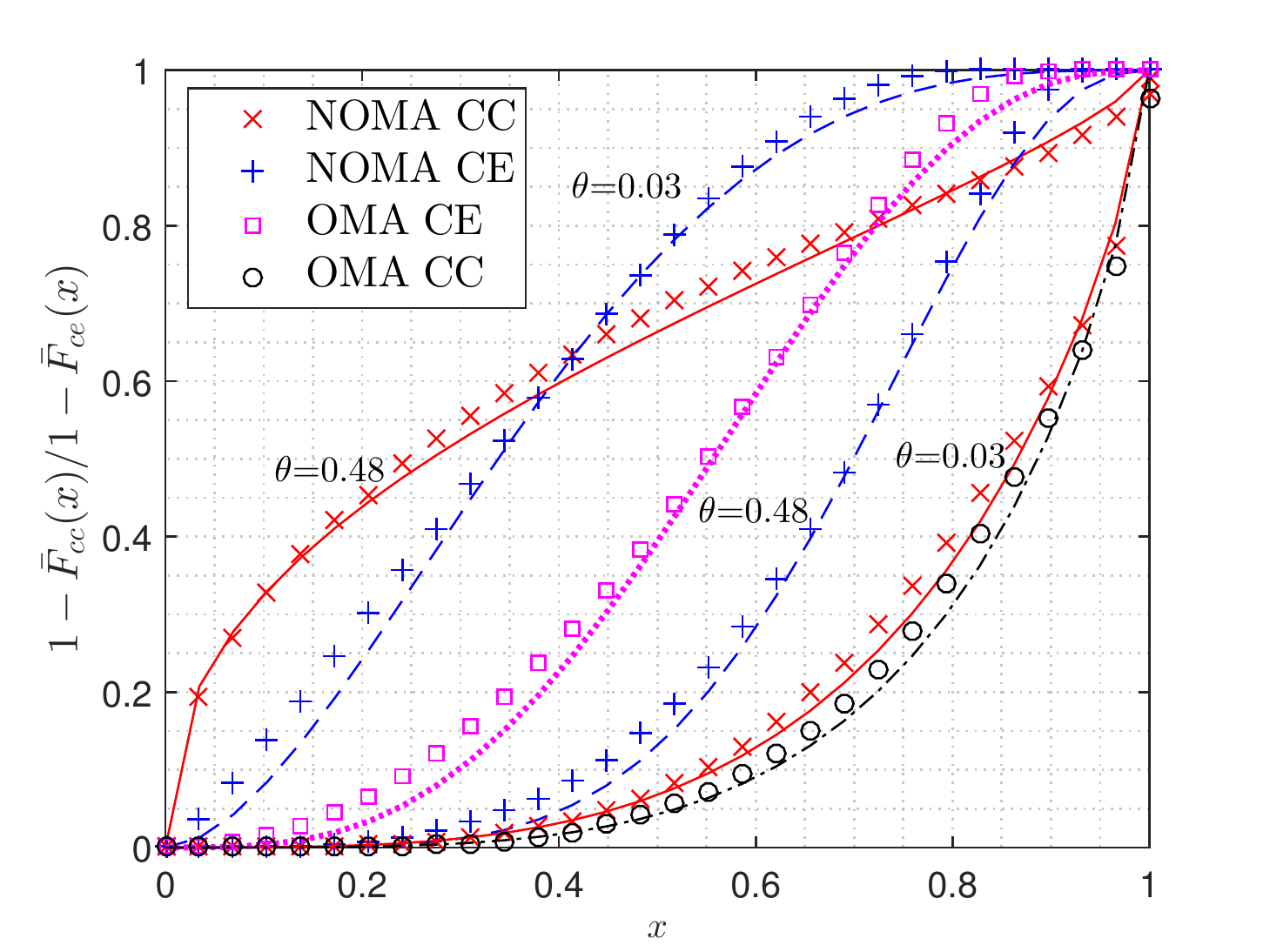} 
\caption{Moments for the CC and CE users under NOMA (Left) and  OMA (Middle), and the  beta approximations (Right) for $\tau=0.7$, $\alpha=4$, $\lambda=1$, and $(\beta_c,\beta_e)=(0,-3)$ dB. The solid and dashed curves correspond to the  analytical results and the markers correspond to the simulation results.}
\label{fig:Moments_NOMA_OMA_BetaApp}
\end{figure*} 
\subsection{Meta Distribution and its Approximation}
\label{subsec:BetaApproximation} 
Using $M_{b}^{c}$ and $M_{b}^{e}$ derived in Theorem \ref{thm:MetaDisMoment} and the Gil-Pelaez's inversion theorem \cite{Gil1951}, the meta distributions for the typical CC and CE users under NOMA can be obtained.
However, the evaluation of meta distributions using the Gil-Pelaez inversion is computationally complex. Thus, similar to \cite{Martin2016Meta},  we approximate the meta distributions of the CC and CE users with the beta distributions by matching the moments as below 
\begin{equation}
\begin{split}
\bar{F}_\text{c}(\chi_c;x)&\approx 1-I(x;\kappa_{1c},\kappa_{2c})\\ \text{ and } \bar{F}_\text{e}(\chi_e;x)&\approx 1-I(x;\kappa_{1e},\kappa_{2e}), 
\end{split}
\label{eq:BetaApp}
\end{equation}
respectively, where $I(x;a,b)$ is the regularized incomplete beta function and $$\kappa_{1s}= \frac{M_1^{s}\kappa_2^{ss}}{1-M_1^{s}} \text{~~and~~} \kappa_{2s}=\frac{(M_1^{s}-M_2^{s})(1-M_1^{s})}{M_2^{s}-(M_1^{s})^2}$$
For $s=\{c,e\}$.
Similarly, the beta approximations for the meta distributions under OMA can be obtained using the moments given in Corollary \ref{cor:OrthogonalTransmission}. We denote the parameters of the beta approximation for the CC (CE) user under OMA by $\tilde{\kappa}_{1c}$ and $\tilde{\kappa}_{1c}$ ($\tilde{\kappa}_{1e}$ and $\tilde{\kappa}_{1e}$).
Fig. \ref{fig:Moments_NOMA_OMA_BetaApp} (Right) shows that the beta distributions closely approximate the meta distributions of the CC and CE users for both NOMA and OMA. The proposed beta approximations can be used for the system-level analysis without having to perform the computationally complex Gil-Pelaez inversion.  {The figure shows that  the percentage of the CE users meeting the link reliability (i.e., conditional success probability) drops with increasing $\theta$ whereas the percentage of CC users achieving the link reliability increases with $\theta$.}
\section{Throughput and Delay Analysis for the CC and CE Users}
\label{sec:ThroughputDelayAnalysis}
In this section, we characterize  the conditional transmission rate (given in \eqref{eq:RateCondPhi_CC_CE}) and the conditional mean delay (given in \eqref{eq:DelayCondPhi_CC_CE}) of the CC/CE users under NRT and RT services, respectively. For this, we first derive the distributions of the CC and CE loads in the following subsection. 
\subsection{Distributions of the CC and CE Loads}
\label{subsec:AreaDistributions}
The CC and CE loads of the typical cell $V_o$, i.e., $N_{oc}$ and $N_{oc}$, depend on the areas of CC and CE regions, i.e., $|V_{oc}|$ and $|V_{oe}|$.  {It is challenging to derive the area distributions of a random set directly. Thus, we first drive the exact first two moments of these areas in the following lemma and then use them to approximate the area distributions of CC and CE regions.} 
\begin{lemma}
\label{lemma:Area_CC_CE_Region}
For a given $\tau$, the mean areas of the CC and CE regions are
{\small \begin{align}
\nbbE[|{V}_{oc}|]&={\tau^2}{\lambda^{-1}}\text{~and~}\nbbE[|{V}_{oe}|]={(1-\tau^2)}{\lambda^{-1}},\label{eq:Mean_CCCERegion}
\end{align} }
respectively,  and second moments of the areas of the CC and CE regions are
{\small \begin{equation}
\begin{split}
\nbbE[|V_{oc}|^2]=4\pi&\int_0^\pi\int_0^\infty\int_0^\infty\exp\left(-\lambda U_3\right)r_1{\rm d}r_1r_2{\rm d}r_2{\rm d}u\\
\text{and}~\nbbE[|V_{oe}|^2]&=4\pi\int_0^\pi\int_0^\infty \ncalF(r_2,u)r_2{\rm d}r_2{\rm d}u,
\end{split}\label{eq:SecondMoment_CCCERegion}
\end{equation}}
 respectively, where
{\small\begin{align*}
\ncalF&(r_2,u)=\int_{\mathtt{D}{(r_2,u)}}\hspace{-.5cm}\left(\exp(-\lambda U_1)\mathbbm{1}_{r_1\leq r_2}+\exp(-\lambda U_2)\mathbbm{1}_{r_2<r_1}\right)r_1{\rm d}r_1\\
&+\int_{\R\setminus\mathtt{D}{(r_2,u)}}\hspace{-.5cm}([\exp(-\lambda {U}_o)-\exp(-\lambda( {U}_1+ {U}_2- {U}_3)]+\\
&~~~~[\exp(-\lambda( {U}_2- {U}_3)-1][\exp(-\lambda {U}_1)-\exp(-\lambda {U}_3)])r_1{\rm d}r_1,
\end{align*}}
{\small $\mathtt{D}(r_2,u)=\{r_1\in\R: d\leq \tau^{-1}|r_1-r_2|\}$,\\ $d=(r_1^2+r_2^2-2r_1r_2\cos(u))^\frac{1}{2}$,  $u_o=u$,\\
 $U_o=U(r_1,r_2,u_o)$, $U_3=U(r_1\tau^{-1},r_2\tau^{-1},u_3)$,\\
$U_1=U(r_1\tau^{-1},r_2,u_1)$ if $r_1\tau^{-1}<d+r_2$ otherwise $U_1=\pi r_1^2\tau^{-2}$, \\
 $U_2=U(r_1,r_2\tau^{-1},u_2)$ if $r_2\tau^{-1}<d+r_1$ otherwise $U_2=\pi r_2^2\tau^{-2}$,\\
  $u_1=\arccos\left((\tau^{-1}-\tau)\frac{r_1}{2r_2}+\tau\cos(u)\right)$,\\
   $u_2=\arccos\left((\tau^{-1}-\tau)\frac{r_2}{2r_1}+\tau\cos(u)\right)$, \\
$u_3=\arccos\left((1-\tau^2)\frac{r_1^2+r_2^2}{2r_1r_2}+\tau^2\cos(u)\right)$,\\ $w(r_1,r_2,u)=\arccos\left(d^{-1}{(r_1-r_2\cos(u))}\right)$, and \\
 $U(r_1,r_2,u)=r_1^2\left(\pi-w(r_1,r_2,u) + \frac{\sin\left(2w(r_1,r_2,u)\right)}{2}\right)$\\
 \hspace*{17mm}$+~r_2^2\left(\pi-w(r_2,r_1,u) + \frac{\sin\left(2w(r_2,r_1,u)\right)}{2}\right)$.}
\end{lemma}
\begin{proof}
Please refer to Appendix C.
\end{proof}
Now,  we approximate the distributions of the CC and CE areas.  {The area of the PV cell follows the gamma distribution \cite{tanemura2003statistical} and the sum of two independent gamma distributed random variables also follows the gamma distribution. Therefore, the gamma distribution is the natural choice for these approximations as the correlation between the CC and CE areas is not too high.} Thus, $\pdf$s of the  CC and CE areas are, respectively, approximated as
\begin{equation}
\begin{split}
f_{|V_{oc}|}(a)&=\frac{\gamma_{1c}^{\gamma_{2c}}}{\Gamma(\gamma_{2c})}a^{\gamma_{2c}-1}\exp(-\gamma_{1c} a)\\
\text{and } f_{|V_{oe}|}(a)&=\frac{\gamma_{1e}^{\gamma_{2e}}}{\Gamma(\gamma_{2e})}a^{\gamma_{2e}-1}\exp(-\gamma_{1e} a),
\end{split}
\label{eq:PDF_CCCEarea}
\end{equation}
 {where for~} $s\in\{c,e\}$ {~we have~}$$\gamma_{2s}=\gamma_{1s}\E[|V_{os}|] \text{~~and~~}\gamma_{1s}=\frac{\E[|V_{os}|]}{\E[|V_{os}|^2]-\E[|V_{os}|]^2}.$$
Fig. \ref{fig:AreaDistribution} provides the visual verification of the gamma approximations given in \eqref{eq:PDF_CCCEarea}. Now, using \eqref{eq:PDF_CCCEarea}, we obtain the distributions of the CC and CE loads in the following lemma.
\begin{figure}[h]
 \centering\vspace{-7mm}
 \includegraphics[width=.48\textwidth]{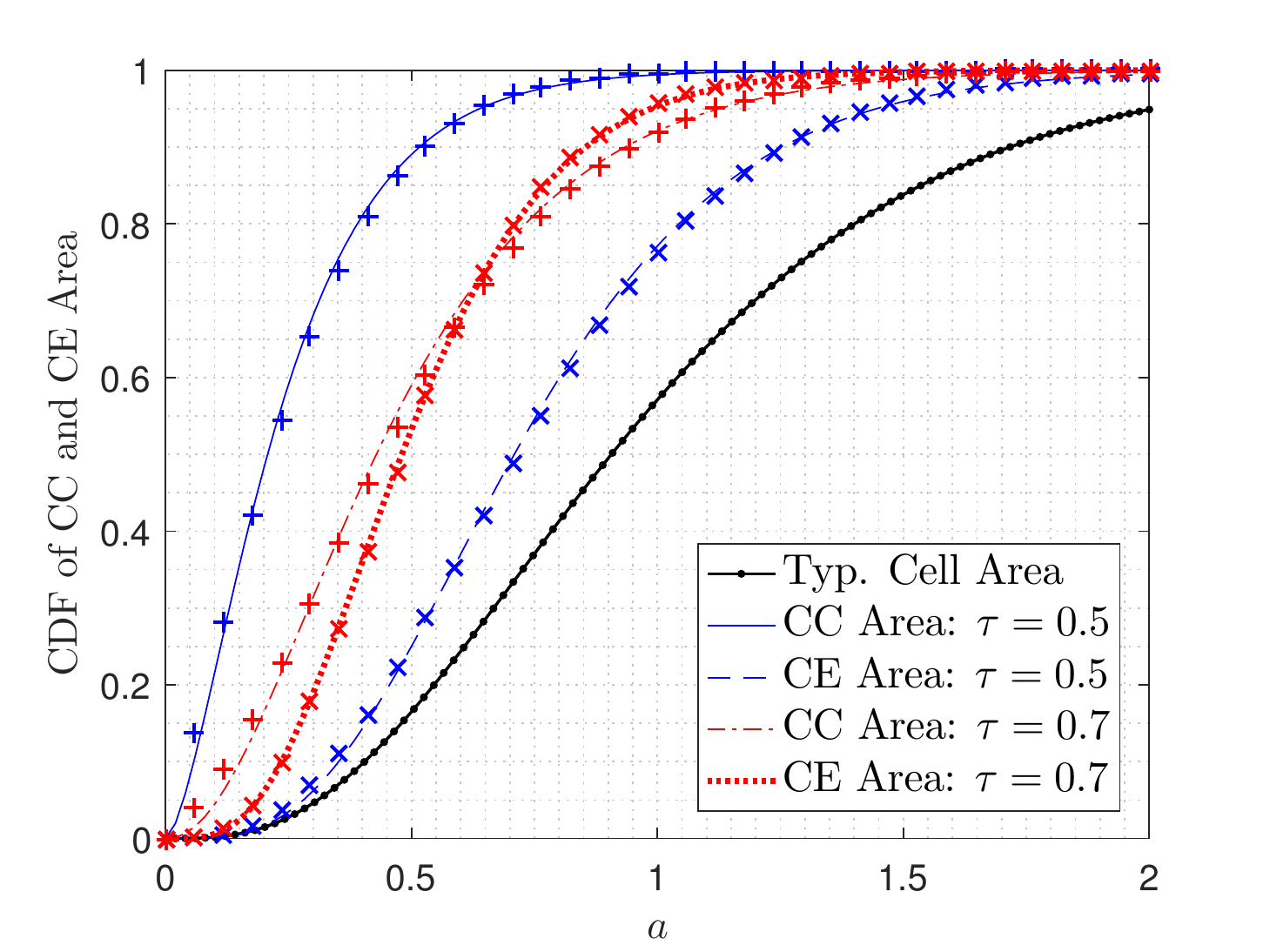}\vspace{-3mm}
\caption{Gamma approximation of the area distributions of the CC and CE regions. The solid and dashed curves correspond to the gamma approximations and the markers correspond to the simulation results.}\vspace{-4mm}
\label{fig:AreaDistribution}
\end{figure}
\begin{lemma}
\label{lemma:CC_CE_Loads}
The probability mass function ($\pmf$) of the number of CC users, i.e., $N_{oc}$, is  $\P[N_{oc}=n]=$
\begin{equation}
\frac{\nu^n\gamma_{1c}^{\gamma_{2c}}}{n!\Gamma(\gamma_{2c})}\int_0^\infty a^{n+\gamma_{2c}-1}\frac{\exp(-(\nu+\gamma_{1c}) a)}{1-\exp(-\nu a)}{\rm d}a,
\label{eq:PMF_Noc}
\end{equation}
where $\gamma_{1c}$ and $\gamma_{2c}$ are given in \eqref{eq:PDF_CCCEarea}. 
The $\pmf$ of the number of CE, i.e., $N_{oe}$, is  $\P[N_{oe}=n]=$
\begin{equation}
\frac{\nu^n\gamma_{1e}^{\gamma_{2e}}}{n!\Gamma(\gamma_{2e})}\int_0^\infty a^{n+\gamma_{2e}-1}\frac{\exp(-(\nu+\gamma_{1e}) a)}{1-\exp(-\nu a)}{\rm d}a,
\label{eq:PMF_Noe}
\end{equation}
where $\gamma_{1e}$ and $\gamma_{2e}$ are given in  \eqref{eq:PDF_CCCEarea}.
\end{lemma}
\begin{proof}
For given $|V_{oc}|$, $N_{oc}$ follows zero-truncated Poisson with mean $\nu|V_{oc}|$. Thus, we have 
\begin{align*}
\P[N_{oc}=n]&=\E_{|V_{oc}|}\left[\P\left[N_{oc}=n\mid|V_{oc}|\right]\right]~\text{for}~n>0.
\end{align*}
 {Now, taking expectation over $\pdf$ of $|V_{oc}|$ given in \eqref{eq:PDF_CCCEarea} will provide the $\pmf$ of $N_{oc}$ as given in \eqref{eq:PMF_Noc}.} Similarly, the $\pmf$ of $N_{oe}$ given in \eqref{eq:PMF_Noe} follows using the $\pdf$ of $|V_{oe}|$ given in \eqref{eq:PDF_CCCEarea}.
\end{proof}
\subsection{Transmission Rates of the CC and CE Users}
\label{subsec:TransmissionRate}
In this subsection, we derive the distributions of the conditional transmission rates of the CC and CE users under random scheduling which are, respectively, defined as 
\begin{equation}
\begin{split}
\ncalR_c(\mathtt{r_{c}};\chi_c)&=\P[R_c(\y,\Phi)\leq \mathtt{r_{c}}]\\
\text{~and~}\ncalR_e(\mathtt{r_{e}};\chi_e)&=\P[R_e(\y,\Phi)\leq \mathtt{r_{e}}],
\end{split}\label{eq:RateOutage_CC_CE}
\end{equation}
where $R_c(\y,\Phi)$ and $R_e(\y,\Phi)$ are given in \eqref{eq:RateCondPhi_CC_CE}.  {Now, using the meta distributions  and the $\pmf$s of the cell loads, the means and the distributions of the conditional transmission rates of the CC and CE users are derived in the following theorem.}
 \begin{theorem}
 \label{thm:TransmissionRate_NOMA}
 The mean transmission rates of the typical CC and CE users under NOMA are
{ \begin{align}
\begin{split}
\bar{R}_{c}(\chi_c) &=\xi_{c} \log_2\left(1+\beta_c\right)M_{1}^{c}(\chi_c)\\
\text{and~} \bar{R}_{e}(\chi_e) &= \xi_{e}\log_2\left(1+\beta_e\right)M_{1}^{e}(\chi_e),
\end{split}
\label{eq:MeanTransmissionRate_CC_CE}
\end{align}} 
respectively, where $M_1^{c}(\chi_c)$ and  $M_1^{e}(\chi_e)$ are given in \eqref{eq:MetaDisMoment_CC_CE},  $s\in\{c,e\}$, $\xi_s=\sum_{n=1}^{\infty}\frac{1}{n}\P[N_{os}=n]$, and  $\P[N_{os}=n]$ is given in Lemma \ref{lemma:CC_CE_Loads}.
The $\cdf$ of the conditional transmission rates of the typical CC and CE users under NOMA are, respectively,
{ \begin{align}
&\ncalR_c(\mathtt{r_{c}};\chi_c)=\nonumber\\
&\E_{N_{oc}}\left[I\left(\min\left(\mathtt{r_{c}} N_{oc}\log_2(1+\beta_c)^{-1},1\right);\kappa_{1c},\kappa_{2c}\right)\right],\label{eq:RateCDF_UpperBound_CC}\\
&\text{and }\ncalR_e(\mathtt{r_{e}};\chi_e)=\nonumber\\
&\E_{N_{oe}}\left[I\left(\min\left(\mathtt{r_{e}} N_{oe}\log_2(1+\beta_e)^{-1},1\right);\kappa_{1e},\kappa_{2e}\right)\right],
\label{eq:RateCDF_UpperBound_CE}
\end{align}}
where  $s\in\{c,e\}$, $\kappa_{1s}$ and $\kappa_{2s}$ are given in \eqref{eq:BetaApp}, and  $\P[N_{os}=n]$ is given in Lemma \ref{lemma:CC_CE_Loads}.
 \end{theorem}
 \begin{proof}
 Please refer to Appendix D.
 \end{proof}
For  OMA, we consider that each BS serves its associated CC and CE users for $\eta$ and $1-\eta$ fractions of time, respectively. Note that for $\eta=\frac{|V_{oc}|}{|V_o|}$, the above OMA scheduling scheme will be almost equivalent to the random scheduling wherein the typical BS randomly schedules one of its associated users in a given time slot. 
\begin{cor}
\label{cor:TransmissionRate_OMA}
 The mean transmission rates of the typical CC and CE users under OMA are
{ \begin{equation}
\begin{split}
\tilde{{R}}_{c}(\beta_c)& =\eta\xi_{c} \log_2\left(1+\beta_c\right)\tilde{M}_{1}^{c}(\beta_c)\\
\text{and~} \tilde{{R}}_{e}(\beta_e)& = (1-\eta)\xi_{e}\log_2\left(1+\beta_e\right)\tilde{M}_{1}^{e}(\beta_e),
\end{split}
\label{eq:MeanTransmissionRate_CC_CE_OMA}
\end{equation}} 
respectively, where $\tilde{M}_1^{c}(\chi_c)$ and $\tilde{M}_1^{e}(\chi_e)$ are given in \eqref{eq:MetaDisMoment_CC_CE_OMA} and, $s\in\{c,e\}$, $\xi_s=\sum_{n=1}^{\infty}\frac{1}{n}\P[N_{os}=n]$, and  $\P[N_{os}=n]$ is given in Lemma \ref{lemma:CC_CE_Loads}.
The $\cdf$ of the conditional transmission rates of the typical CC and CE users under OMA respectively are
{ \begin{align}
&\tilde{\ncalR}_c(\mathtt{r_{c}};\beta_c)=\nonumber\\
&\E_{N_{oc}}\left[I\left(\min\left( \frac{\mathtt{r_{c}}N_{oc}\eta ^{-1}}{\log_2(1+\beta_c)},1\right);\tilde{\kappa}_{1c},\tilde{\kappa}_{2c}\right)\right],\\
&\text{and }\tilde{\ncalR}_e(\mathtt{r_{e}};\beta_e)=\nonumber\\
&\E_{N_{oe}}\left[I\left(\min\left( \frac{\mathtt{r_{e}}N_{oe}(1-\eta)^{-1}}{\log_2(1+\beta_e)},1\right);\tilde{\kappa}_{1e},\tilde{\kappa}_{2e}\right)\right],
\end{align}}
where $s\in\{c,e\}$, $\tilde{\kappa}_{1s}$ and $\tilde{\kappa}_{2s}$ are given in Section \ref{subsec:BetaApproximation}, and $\P[N_{os}=n]$ is given in Lemma \ref{lemma:CC_CE_Loads}.
\end{cor}
{\begin{proof}
In OMA case, the transmission rates of the CC and CE users, for given $\y$ and $\Phi$, are $\frac{\eta}{N_{oc}}\log_2(1+\beta_c)\E\left[\mathbbm{1}_{\tilde{\ncalE_c}}(\sir_{c})\mid\y,\Phi\right]$ and $\frac{1-\eta}{N_{oe}}\log_2(1+\beta_c)\E\left[\mathbbm{1}_{\tilde{\ncalE_e}}(\sir_{e})\mid\y,\Phi\right]$, respectively. Hence, further following the steps in Appendix D, we complete the proof. 
\end{proof}}
\subsection{Delay Analysis of the CC and CE Users}
\label{subsec:PacketDelay}
This subsection analyzes the delay performance of the typical CC/CE user for the given RT service under the NOMA setup described in Section \ref{subsec:SchedulinThroughputDelay}. Because of the assumption of saturated queues at the interfering BSs, the meta distributions derived in Section \ref{sec:MetaDistributions_NOMA_OMA} can be directly used to analyze the upper bound of the delay performance of the typical CC and CE user (see Section \ref{subsec:SchedulinThroughputDelay}). That said, the conditional packet transmission rate of the typical CC/CE user is the product of its scheduling probability and success probability as stated in \eqref{eq:DelayCondPhi_CC_CE}.  {It may be noted that the successful transmission events  across the  time slots are independent for given $\y$ and $\Phi$.  Hence, the service times of packets of the typical CC/CE user at $\y$ given $\Phi$ are i.i.d. and follow a geometric  distribution with parameter $\mu_c(\y,\Phi)$/$\mu_e(\y,\Phi)$.}  Besides, the packet arrives in each time slot as per the Bernoulli process with mean $\varrho_c$/$\varrho_e$. Thus, the queue of the typical CC/CE user can be modeled as the Geo/Geo/1 queue. {The upper bounds of the conditional mean delays of the typical CC and CE users becomes \cite{atencia2004discrete}}
{ \begin{align}
\begin{split}
D_c(\y,\Phi)&=\frac{1-\varrho_c}{\mu_c(\y,\Phi)-\varrho_c}\mathbbm{1}_{\mu_c(\y,\Phi)>\varrho_c}\\
\text{ and } D_e(\y,\Phi)&=\frac{1-\varrho_e}{\mu_e(\y,\Phi)-\varrho_e}\mathbbm{1}_{\mu_e(\y,\Phi)>\varrho_e},
\end{split}\label{eq:MeanConditionalDelay_CC_CE}
\end{align}}
respectively. Let $\mathtt{t_{c}}$ and $\mathtt{t_e}$ are mean delay thresholds of the CC and CE users, respectively.
 \begin{theorem}
 \label{thm:CDF_CondMeanDelay}
 The complementary $\cdf$ ($\ccdf$) of the conditional mean delays of the typical CC and  CE users under NOMA with random scheduling are upper bounded respectively by
  \begin{align}
&\ncalD_c(\mathtt{t_{c}};\chi_c)=\nonumber\\
& \E_{N_{oc}}\left[I\left(\min\left(N_{oc}\left(\frac{1-\varrho_c}{\mathtt{t_{c}}}+\varrho_c\right),1\right);\kappa_{1c},\kappa_{2c}\right)\right],\label{eq:DelayCDF_UpperBound_CC}\\
&\text{and }\ncalD_e(\mathtt{t_{e}};\chi_e)=\nonumber\\
&\E_{N_{oe}}\left[I\left(\min\left(N_{oe}\left(\frac{1-\varrho_e}{\mathtt{t_{e}}}+\varrho_e\right),1\right);\kappa_{1e},\kappa_{2e}\right)\right]
\label{eq:DelayCDF_UpperBound_CE}
\end{align}
where $s\in\{c,e\}$, $\kappa_{1s}$ and $\kappa_{2s}$ are given in \eqref{eq:BetaApp}, and  $\P[N_{oc}=n]$ is given by Lemma \ref{lemma:CC_CE_Loads}.
 \end{theorem}
 \begin{proof}  
Using Assumption 1 along with \eqref{eq:DelayCondPhi_CC_CE} and  \eqref{eq:MeanConditionalDelay_CC_CE}, the $\cdf$ of $D_c(\y,\Phi)$ becomes $ \P[D_c(\y,\Phi)< \mathtt{t_{c}}]$
 { \begin{align*}
&=\P\left[\mu_c(\y,\Phi)>\frac{1-\varrho_c}{\mathtt{t_c}}+\varrho_c,\mu_c(\y,\Phi)>\varrho_c\right],\\
&= \E_{N_{oc}}\left[\P\left(p_c(\beta_c,\beta_e\mid\y,\Phi)> N_{oc}\left(\frac{1-\varrho_c}{\mathtt{t_{c}}}+\varrho_c\right)\mid N_{oc}\right)\right]
\end{align*}}  
{Further, using the meta distribution for the system model discussed in Section \ref{subsec:SchedulinThroughputDelay}, the $\ccdf$ of the upper bounded mean delay of the CC users $D_c(\y,\Phi)$ is given by \eqref{eq:DelayCDF_UpperBound_CC}. Similarly,  the $\ccdf$ of the upper bounded mean delay of CE users $D_e(\y,\Phi)$  is obtained as given in \eqref{eq:DelayCDF_UpperBound_CE}.} 
 \end{proof}
For OMA, the mean delay of the typical CC and CE users at $\y$ conditioned on $\Phi$ becomes
{ \begin{align}
\begin{split}
\tilde{\mu}_{c}(\y,\Phi) &=\frac{\eta\E\left[\mathbbm{1}_{\tilde{\ncalE}_c}(\sir_{c})\mid\y,\Phi\right]}{N_{oc}}\\
\text{ and } \tilde{\mu}_{e}(\y,\Phi) &= \frac{(1-\eta)\E\left[\mathbbm{1}_{\tilde{\ncalE}_e}(\sir_{e})\mid\y,\Phi\right]}{N_{oe}},
\end{split}
\label{eq:DelayCondPhi_CC_CE_OMA}
\end{align}}
The following corollary presents the upper bounded $\ccdf$s of mean delays for the OMA case.
 \begin{cor}
 \label{cor:CDF_CondMeanDelay}
  The $\ccdf$ of the mean delays of the typical CC and  CE users under OMA with random scheduling are upper bounded respectively by
 \begin{align}
&\tilde{\ncalD}_c(\mathtt{t_{c}};\beta_c)=\nonumber\\
& \E_{N_{oc}}\left[I\left(\min\left( \frac{N_{oc}}{\eta}\left(\frac{1-\varrho_c}{\mathtt{t_{c}}}+\varrho_c\right),1\right);\tilde{\kappa}_{1c},\tilde{\kappa}_{2c}\right)\right],\label{eq:DelayCDF_UpperBound_CC_OMA}\\
&\text{and }\tilde{\ncalD}_e(\mathtt{t_{e}};\beta_e)=\nonumber\\
& \E_{N_{oe}}\left[I\left(\min\left(\frac{N_{oe}}{1-\eta}\left(\frac{1-\varrho_e}{\mathtt{t_{e}}}+\varrho_e\right),1\right);\tilde{\kappa}_{1e},\tilde{\kappa}_{2e}\right)\right],\label{eq:DelayCDF_UpperBound_CE_OMA}
\end{align}
where $s\in\{c,e\}$, $\tilde{\kappa}_{1s}$ and $\tilde{\kappa}_{2s}$ are given in Section \ref{subsec:BetaApproximation}, and  $\P[N_{os}=n]$ is given in Lemma \ref{lemma:CC_CE_Loads}.
 \end{cor}
{ \begin{proof}
Using the successful packet transmission rates $\tilde{\mu}_{c}(\y,\Phi)$ and $\tilde{\mu}_{e}(\y,\Phi)$ given in \eqref{eq:DelayCondPhi_CC_CE_OMA} and further following the proof of Theorem \ref{thm:CDF_CondMeanDelay}, we obtained the upper bounds on $\ccdf$s of mean delays of the CC and CE users under OMA as given in \eqref{eq:DelayCDF_UpperBound_CC_OMA} and \eqref{eq:DelayCDF_UpperBound_CE_OMA}, respectively.
 \end{proof}}
\section{Resource Allocation and Performance Comparison}
\label{sec:CSEoptimization}
{In this section, we focus on the RA for maximizing the network performance under  both NRT and RT services while meeting the QoS constraints of the CC and CE users. 
For NRT services, we focus on the maximization of $\csr$ such that the minimum transmission rates of the CC and CE users are ensured. However, for RT services, we consider the maximization of $\sec$ such that  CC and CE services with minimum arrival rates can be supported, and their corresponding packet transmission delays are also bounded. First, we develop an efficient method to obtain near-optimal RA for the NRT services in the following subsection. }
{\subsection{RA under NRT services}
\label{subsec:RA_NRT}
Using the success probabilities of the CC and CE users,  $\csr$s  for the NOMA and OMA systems can be, respectively, obtained as 
\begin{align}
\csr_\mathtt{NOMA}=&\mathtt{B}\log_2(1+\beta_c)M_1^c(\chi_c)\nonumber\\
&+\mathtt{B}\log_2(1+\beta_e)M_1^e(\chi_e),\\
\text{and}~ \csr_\mathtt{OMA}=&\eta\mathtt{B}\log_2(1+\beta_c)\tilde{M}_1^c(\beta_c)\nonumber\\
&+(1-\eta)\mathtt{B}\log_2(1+\beta_e)\tilde{M}_1^e(\beta_e).
\end{align}
The RA formulation for the NRT services is as follows.}

{$\bullet$ $\mathcal{P}_1$- $\csr$ maximization subject to the minimum mean rates of the CC and CE users.
\begin{align*}
\text{NOMA:}\  
\hspace{-0mm}\begin{split}
\max~&\csr_\mathtt{NOMA}\\
\text{s.t.}~&0<\theta<1\\
&\bar R_c(\chi_c)\geq \mathtt{R_c}\\
&\bar R_e(\chi_e)\geq \mathtt{R_e},
\end{split}
\text{~and~OMA:}\  
\hspace{-3mm}\begin{split}
\max~&\csr_\mathtt{OMA}\\
\text{s.t.}~&0<\eta<1\\
&\tilde{R}_c(\beta_c)\geq \mathtt{R_c}\\
&\tilde{R}_e(\beta_e)\geq \mathtt{R_e}.
\end{split}
\end{align*}}
{\subsubsection{ Near-optimal RA for $\ncalP_1$-NOMA}  {It is difficult to obtain an exact optimal RA to $\ncalP_1$-NOMA as it does not fall in the standard convex-optimization framework. Therefore, we present an efficient method to obtain a near-optimal solution based on the insights obtained from the NOMA analysis.} It is natural to allocate the remaining power to the CC user after achieving the minimum transmission rate of the CE user in order to maximize the $\csr$. Therefore, we consider maximizing the mean transmission rate of the CC user under the constraints of the minimum transmission rates of the CC and CE users. 
One can easily see that $\theta\leq\theta_{\mathtt{NC}}=(1+\beta_e)^{-1}$ is the necessary condition for $\sir_e\geq\beta_e$. 
From \eqref{eq:SuccessEevent_CC}, we note that the success probability of the CC user increases with the decrease of $\chi_c$. The success probability (thus, the mean transmission rate) of the CC user is an increasing and decreasing functions of $\theta$ for $0<\theta\leq\hat\theta$ and $\hat\theta<\theta\leq \theta_{\mathtt{NC}}$, respectively, where 
\begin{equation}
\hat\theta=\underset{{0<\theta\leq \theta_{\mathtt{NC}}}}{\mathrm{argmin}}~\chi_c=\left\{\theta:\frac{\beta_c}{\theta}=\frac{\beta_e}{1-\theta(1+\beta_e)}\right\}.
\label{eq:theta_critical}
\end{equation} 
From the above, we know that there are two solutions (if any exists), denoted by $\theta_{lc}$ and $\theta_{uc}$, to $\bar{R}_c(\chi_c)=\mathtt{R_c}$  such that $\theta_{lc}\leq \hat\theta$ and $\theta_{uc}\geq\hat\theta$. Besides,  the transmission rate of the CE user is a  non-increasing function of $\theta$. 
Let $\theta_e$ be the solution of $\bar{R}_e(\chi_e)=\mathtt{R_e}$. Hence, from the above discussion, the optimal allocation becomes $\theta^*=\min\{\theta_e,\hat\theta\}$ if $\theta_{lc}$ and $\theta_{e}$ exist such that $\theta_e\geq\theta_{lc}$. }
{\subsubsection{ Optimal RA for $\ncalP_1$-OMA} For $\beta_c>\beta_e$ and $\tilde{M}_1^c(\beta_c)>\tilde{M}_1^e(\beta_e)$, one can easily infer that the $\mathtt{CSE_{OMA}}$ and mean transmission rate of the CC user are monotonically increasing functions of $\eta$. In contrast, the  mean transmission rate of the CE user is a monotonically decreasing function of $\eta$. 
Therefore, it is straightforward to choose the optimal solution $\eta^*=\eta_e$ for OMA if $\eta_c\leq\eta_e$ where $\eta_c$ and $\eta_e$ are the solutions of $\tilde{R}_c(\chi_c)=\mathtt{R_c}$ and $\tilde{R}_e(\chi_e)=\mathtt{R_e}$, respectively.
Note that there is no feasible solution if $\eta_c>\eta_e$.} 
{\subsection{RA under RT services}
\label{subsec:RA_RT}
To evaluate the $\ec$, we define the delay QoS constraint for the CC/CE users as the probability that the delay outage is below $\mathtt{O}_c$/$\mathtt{O}_e$ for the given mean delay threshold $\mathtt{t_c}$/$\mathtt{t_e}$.
Therefore, using the upper bound of the mean delay outage and the fact that the mean delay outage is a monotonically increasing function of the arrival rate, the lower bounds of $\ec$s  for the CC and CE services under NOMA and OMA become
 \begin{align}
 \ec_\mathtt{NOMA}^s&=\{\varrho_s\in\R_+:\ncalD_s(\mathtt{t_s},\chi_s)=\mathtt{O}_s\}\\
 \text{ and } \ec_\mathtt{OMA}^s&=\{\varrho_s\in\R_+:\tilde{\ncalD}_s(\mathtt{t_s},\beta_s)=\mathtt{O}_s\},
\end{align}
respectively, where $s=\{c,e\}$. The RAs  for RT services are formulated  as below.}

 $\bullet$ $\mathcal{P}_2$- $\sec$ maximization subject to the minimum $\ec$s for the CC and CE users. \\
\begin{align*}
\text{NOMA:}
\begin{split}
\max~&\ec_\mathtt{NOMA}^c+\ec_\mathtt{NOMA}^e\\
\text{s.t.}~&0<\theta<1\\
&\ec_\mathtt{NOMA}^c\geq \bar{\varrho}_c\\
&\ec_\mathtt{NOMA}^e\geq \bar{\varrho}_e,
\end{split}
\text{~and~OMA:}
\hspace{-3mm}\begin{split}
\max~&\ec_\mathtt{OMA}^c+\ec_\mathtt{OMA}^e\\
\text{s.t.}~&0<\eta<1\\
&\ec_\mathtt{OMA}^c\geq \bar{\varrho}_c\\
&\ec_\mathtt{OMA}^e\geq \bar{\varrho}_e.
\end{split}
\end{align*}
The $\ec$ constraints of RA formulation $\ncalP_2$ can facilitate  RT services for the CC and CE users requiring minimum packet arrival rates $\tilde{\varrho}_c$ and $\tilde{\varrho}_e$, respectively. At the same time, this formulation also ensures that the  outage probability of the mean delays  for the CC (CE) user for the given delay threshold $\mathtt{t_c}$ ($\mathtt{t_e}$) is below $\mathtt{O}_c$  ($\mathtt{O}_e$).   The local mean packet delay (i.e., the mean number of slots required for the successful delivery of packet) is equal to the first inverse moment of the conditional success probability \cite{Haenggi2013local}. It is therefore natural that the $\ec$ is higher for the lower inverse moment of the conditional success probability. In fact, it will be evident in section \ref{sec:NumericalResults} that the outages of delay and transmission rate follow similar trends w.r.t. the power allocation $\theta$ under NOMA or the time allocation $\eta$ under OMA. Hence, we can infer that the $\ec$s of CC and CE users also behave similar to their transmission rates. Hence, the maximization of $\sec$ is similar to the maximization of $\csr$. 
{\subsubsection{Near-optimal RA for $P_2$-NOMA} 
Since $\ncalP_2$-NOMA is an non-convex-optimization problem, we adopt the similar approach developed in subsection \ref{subsec:RA_NRT} to obtain its near-optimal solution. 
As already implied above, the $\ec$ of CC users is an increasing and decreasing functions of $\theta$ for $0<\theta\leq\hat\theta$ and $\hat\theta<\theta\leq\theta_{\mathtt{NC}}$, respectively. However, the $\ec$ for the CE users is a decreasing function of $\theta$. Let $\tilde\theta_{lc}$ and $\tilde{\theta}_{uc}$ be the two solutions of $\mathtt{EC}_\mathtt{NOMA}^c=\bar\varrho_c$  such that $\tilde\theta_{lc}\leq \hat\theta$ and $\tilde\theta_{uc}\geq\hat\theta$, and $\tilde\theta_e$ is the solution of $\mathtt{EC}_\mathtt{NOMA}^e=\bar\varrho_e$. Hence, the optimal power allocation becomes $\tilde{\theta}^*=\min(\tilde\theta_e,\hat{\theta})$ if $\tilde\theta_{lc}$ and $\tilde\theta_e$ exist such that $\tilde\theta_e\geq\tilde\theta_{lc}$.}
\subsubsection{Optimal RA for $\ncalP_2$-OMA} Note that $\mathtt{EC}_\mathtt{OMA}^c$ increases with $\eta$ as the scheduling probability for the CC users increases with $\eta$. However, $\mathtt{EC}_\mathtt{OMA}^e$ follows exactly the opposite trend. Besides, allocating more transmission time to the CC users obviously results in higher $\sec$ because of $\tilde{M}_1^c(\beta_c)>\tilde{M}_1^e(\beta_e)$. Therefore, we can obtain the optimal solution $\tilde{\eta}^*=\tilde{\eta}_e$ for OMA if $\tilde\eta_c\leq \tilde\eta_e$ where $\tilde\eta_c$ and $\tilde\eta_e$ are the solutions of $\mathtt{EC}_\mathtt{OMA}^c=\bar{\varrho}_c$ and $\mathtt{EC}_\mathtt{OMA}^e=\bar{\varrho}_e$, respectively.
\begin{figure*}
 \centering
\includegraphics[width=.45\textwidth]{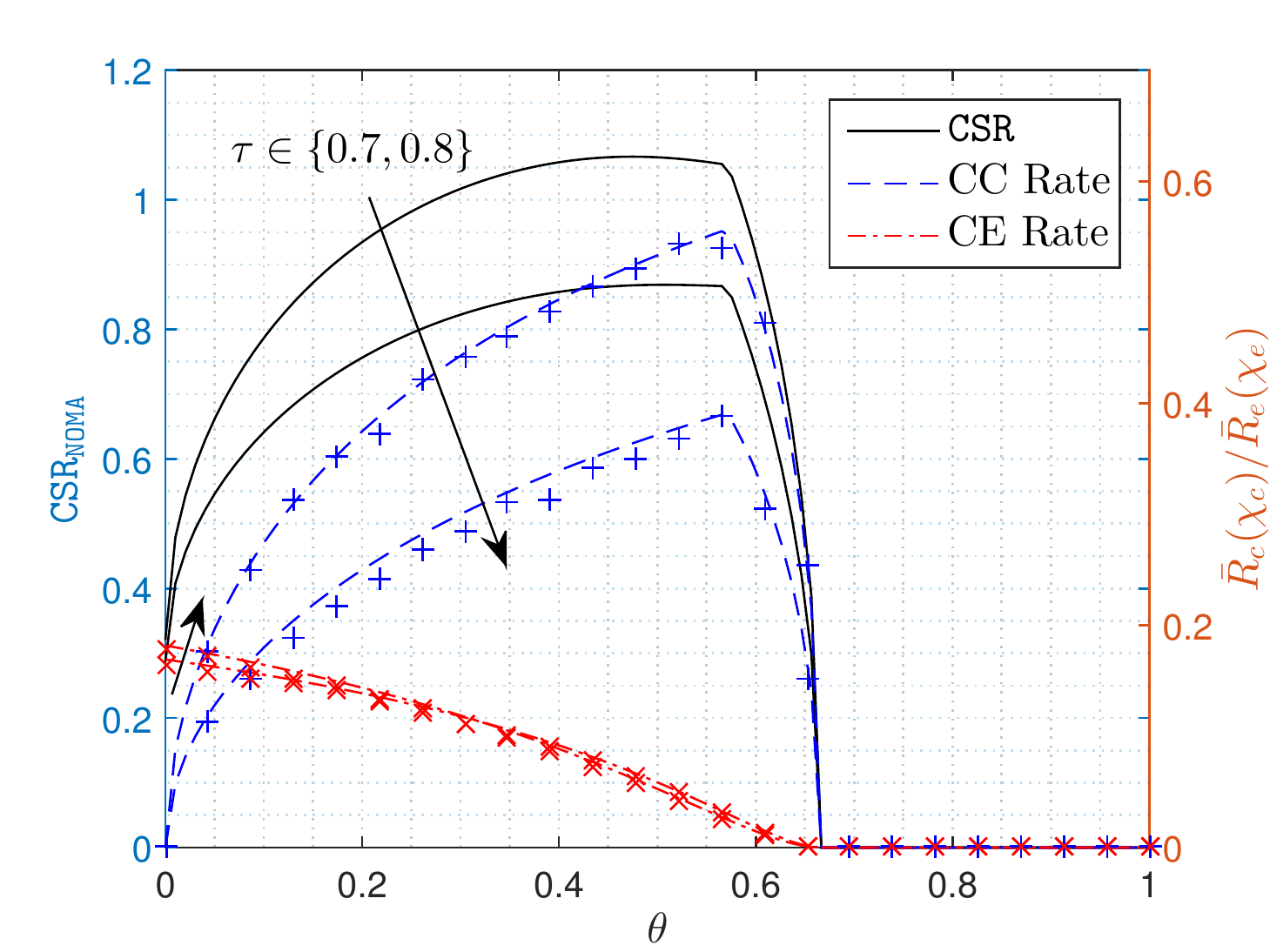}
  \includegraphics[width=.45\textwidth]{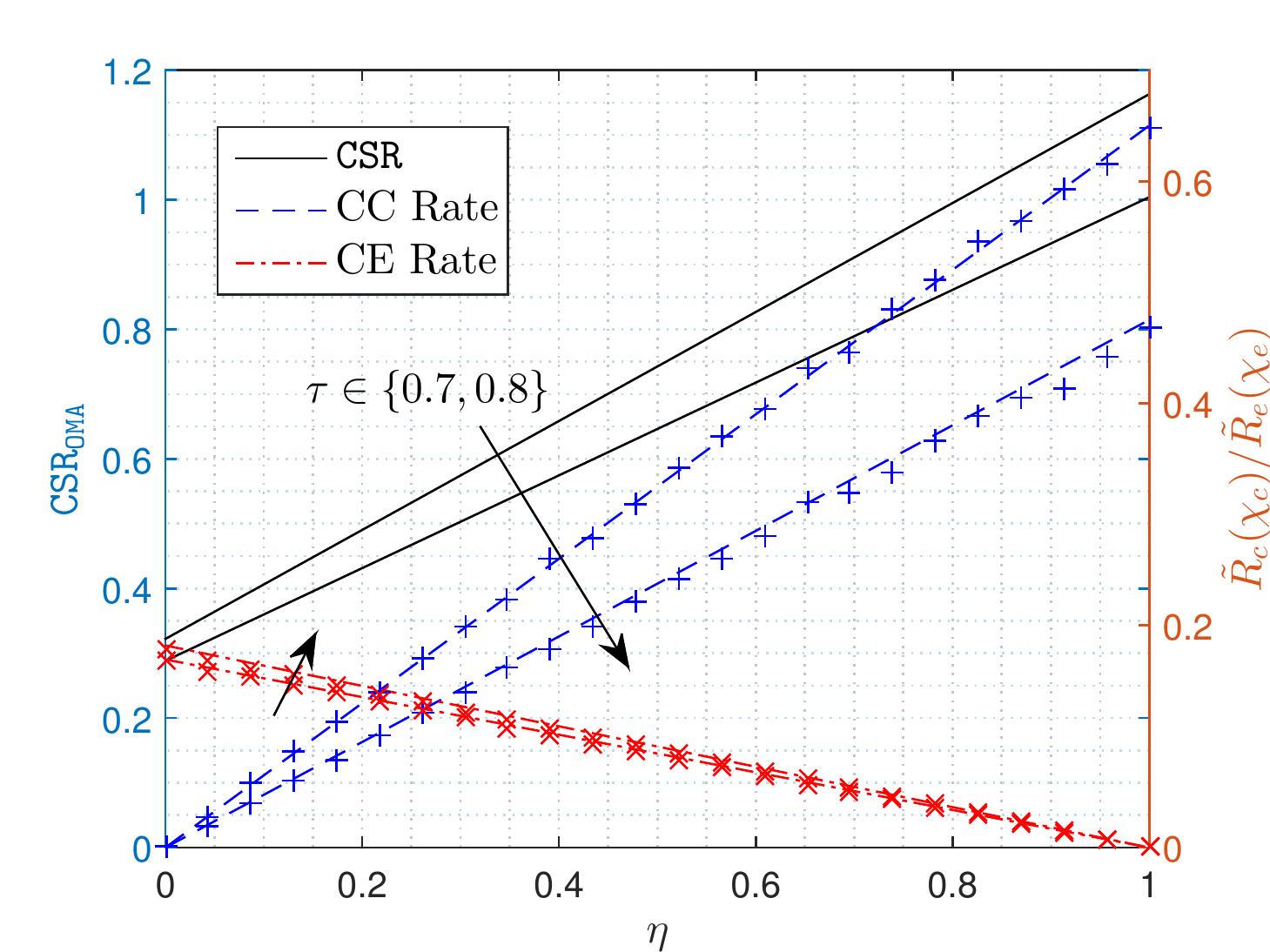}  
\caption{$\csr$ and mean transmission rates of CC and CE users under NOMA (Left) and OMA (Right).  The solid and dashed curves correspond to the analytical results, and the markers correspond to the simulation results.}
\label{fig:MeanThrouput_CSE}
\end{figure*}
{\subsection{Performance gain of NOMA}
  {Due to the intra-cell interference,  the successful transmission probabilities of both the CC and CE users drop in the NOMA system compared to those in OMA system.} However, concurrent transmissions increase their scheduling probabilities in NOMA as compared to those of the OMA system. Since $\tilde{M}_1^c(\beta_c)={M}_1^c(\beta_c)$ and $\tilde{M}_1^e(\beta_e)={M}_1^c(\beta_e)$, the performance gain in the transmission rates of CC and CE users are, respectively, given by
\begin{equation}
\mathtt{g}_c=\frac{M_1^c(\chi_c)}{\eta M_1^c(\beta_c)} ~\text{and}~ \mathtt{g}_e=\frac{M_1^e(\chi_e)}{(1-\eta) M_1^e(\beta_e)}.
\label{eq:NOMA_Gain}
\end{equation}
Thus, the transmission gain is an increasing function of $\theta$ in the interval $(0,\hat\theta]$ for the CC users, whereas  it is a decreasing function of $\theta$ for the CE users.
This implies that there is a performance trade-off between the transmission gains for the CC and CE users.
For a given $\eta$, we have $\mathtt{g}_c>1$ for a set $\Theta_c(\eta)=\{\theta\in[0,\hat\theta]: \frac{M_1^c(\chi_c)}{ M_1^c(\beta_c)} >\eta\}$ and $\mathtt{g}_e>1$ for a set $\Theta_e(\eta)=\{\theta\in[0,\hat\theta]: \frac{M_1^e(\chi_e)}{ M_1^c(\beta_e)} >1-\eta\}$.  Therefore, NOMA is beneficial for both CC and CE users compared to OMA only if $\Theta_c(\eta)\cap\Theta_e(\eta)\neq\emptyset$. Otherwise, at least one of these two types of users will underperform in NOMA  compared to OMA.  It is difficult to analytically show that the intersection of $\Theta_c(\eta)$ and $\Theta(\eta)$ is non-empty for a given $\eta$. That said,  it will be evident in Section \ref{sec:NumericalResults} that this condition does not hold only for higher values of $\eta$.}
\section{Numerical Results and Discussion}
\label{sec:NumericalResults}
In this section, we first verify the accuracy of analytical results by comparing them with the simulation results obtained through Monte Carlo simulations. 
Next, we discuss the performance trends of the achievable $\csr$, and the transmission rates and mean delays of the CC and CE users. Further, we also compare the performances of these  metrics for the NOMA and OMA systems. 
For this, we consider $\lambda=1$, $\nu=5$, $\alpha=4$, $\tau=0.7$, $(\beta_c,\beta_e)=(3,-3)$ in dB, $(\mathtt{r_c},\mathtt{r_e})=(0.1,0.05)$, $(\varrho_c,\varrho_e)=(0.05,0.05)$, and $(\mathtt{t_c},\mathtt{t_e})=(20,30)$,  unless mentioned otherwise. 

Fig. \ref{fig:MeanThrouput_CSE} depicts that the mean transmission rates of the CC and CE users closely match with the simulation results for both the NOMA and OMA systems. The curves correspond to the analytical results, whereas the  markers correspond to the simulation results. The figure gives the visual verifications of the  trends of transmission rate functions discussed in Section \ref{sec:CSEoptimization}. 
It can be seen that the $\csr$ and the transmission rates are zero for $\theta>\theta_{\mathtt{NC}}\approx0.66$ which  verifies the necessary condition  for NOMA operation. In addition, we can see that the transmission rate of the CE user monotonically decreases with $\theta\leq\theta_{\mathtt{NC}}$, whereas the transmission rate of the CC user is a monotonically increasing and decreasing function of $\theta$ for $0<\theta\leq \hat\theta$ and $\hat\theta<\theta\leq\theta_{\mathtt{NC}}$, respectively, where $\hat\theta=0.5$ (see \eqref{eq:theta_critical}).

\begin{figure}[h]
 \centering\vspace{-3mm}
  \includegraphics[width=.45\textwidth]{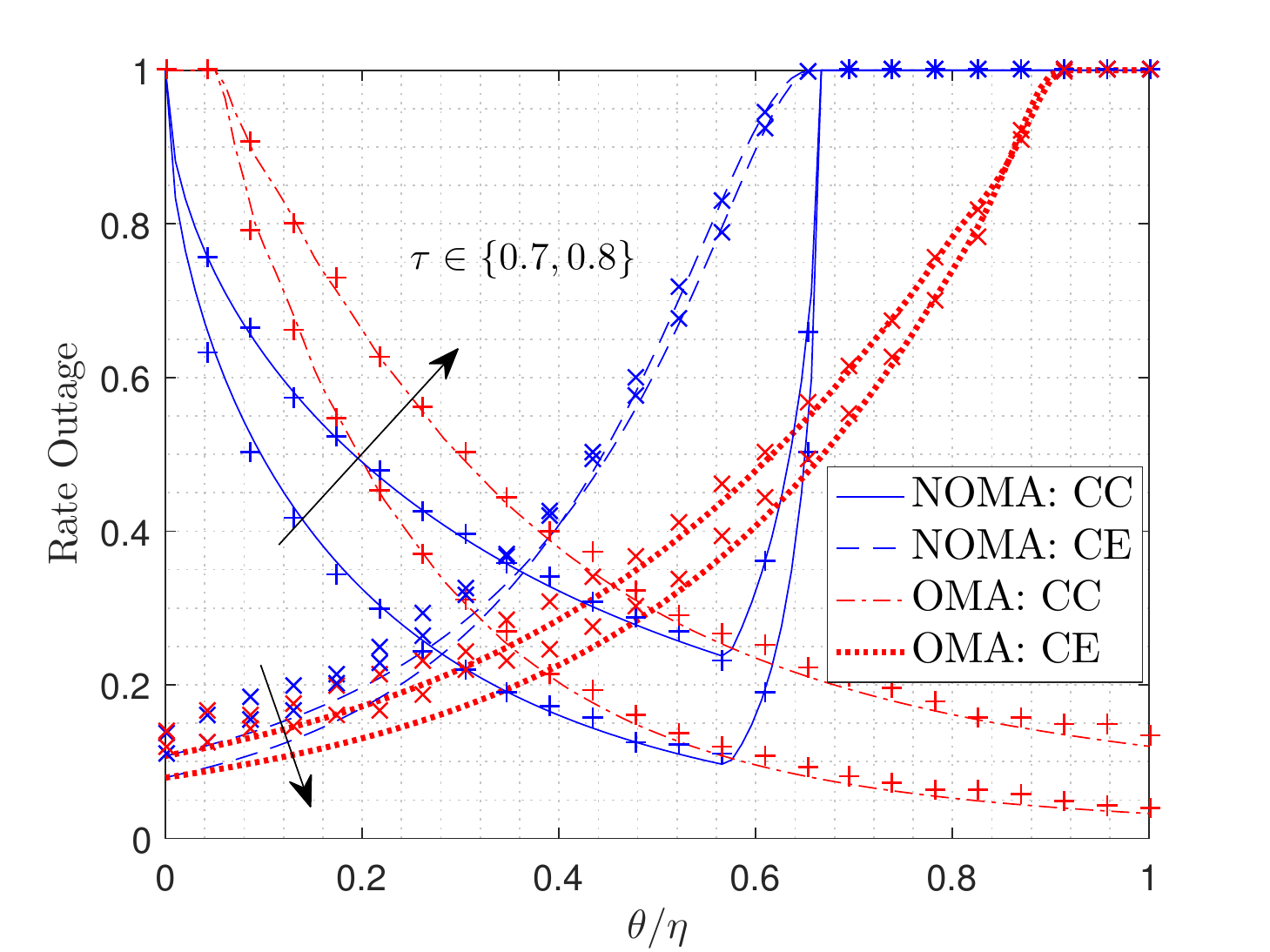}
    \includegraphics[width=.45\textwidth]{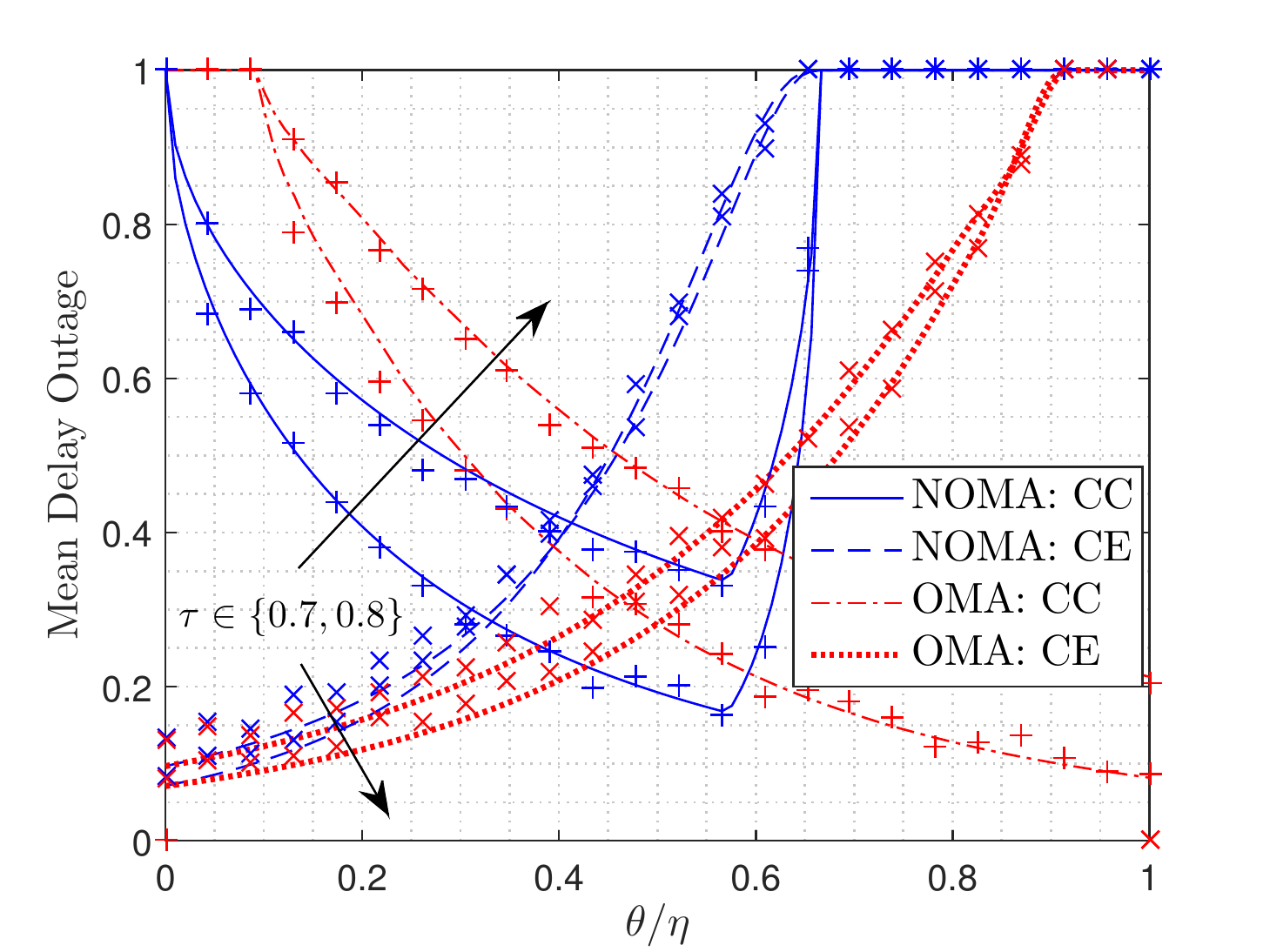}\vspace{-.2cm}
\caption{The outage probabilities of transmission rates (Left) and mean delays (Right) of the CC and CE users. The solid and dashed curves correspond to the analytical results and the markers correspond to the simulation results.}\vspace{-2mm}\label{fig:Rate_Delay_Outage}
\end{figure}
\begin{figure*}
 \centering
 \hspace{-.3cm}\includegraphics[width=.35\textwidth]{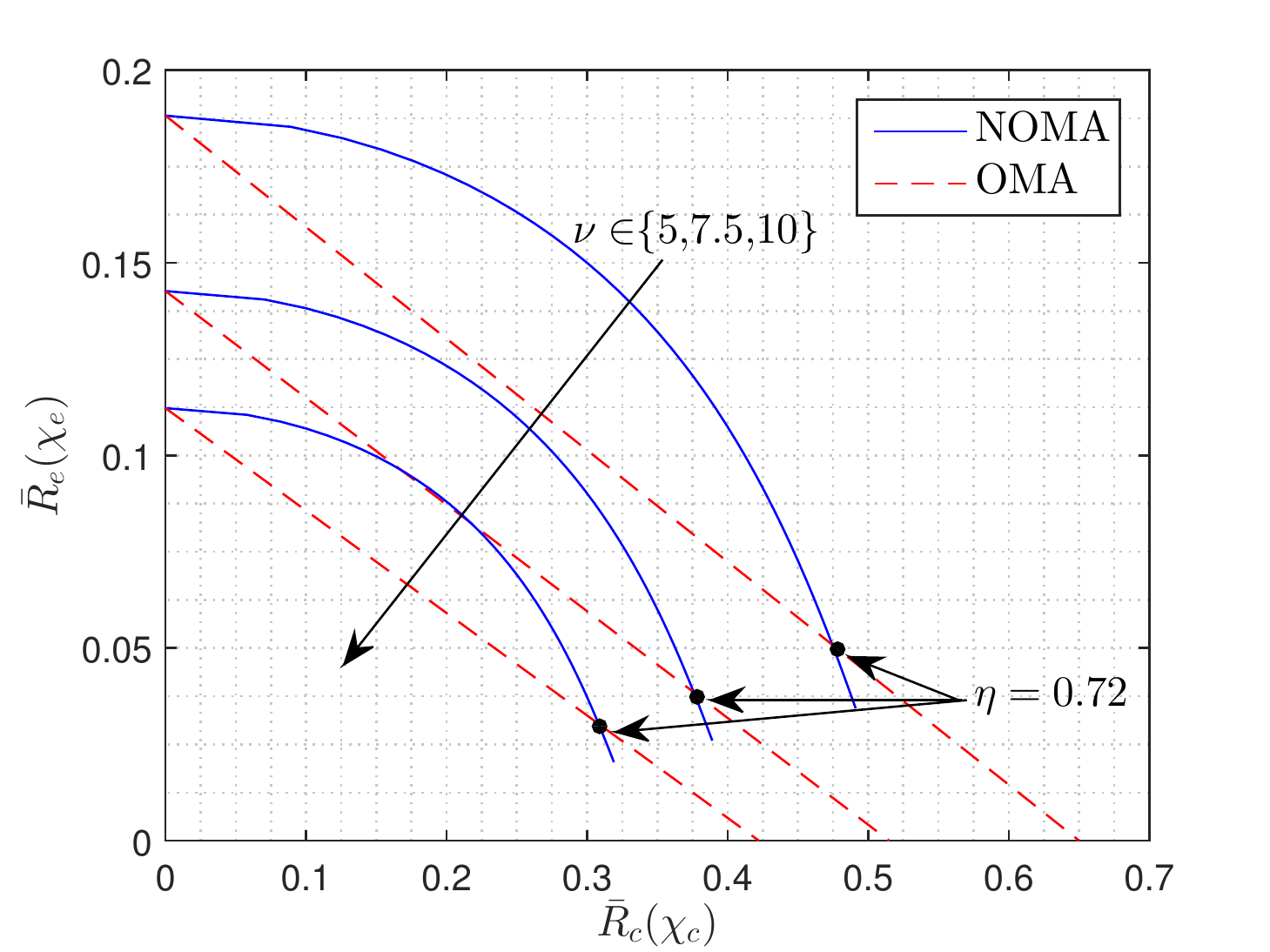} \hspace{-.7cm} 
 \includegraphics[width=.35\textwidth]{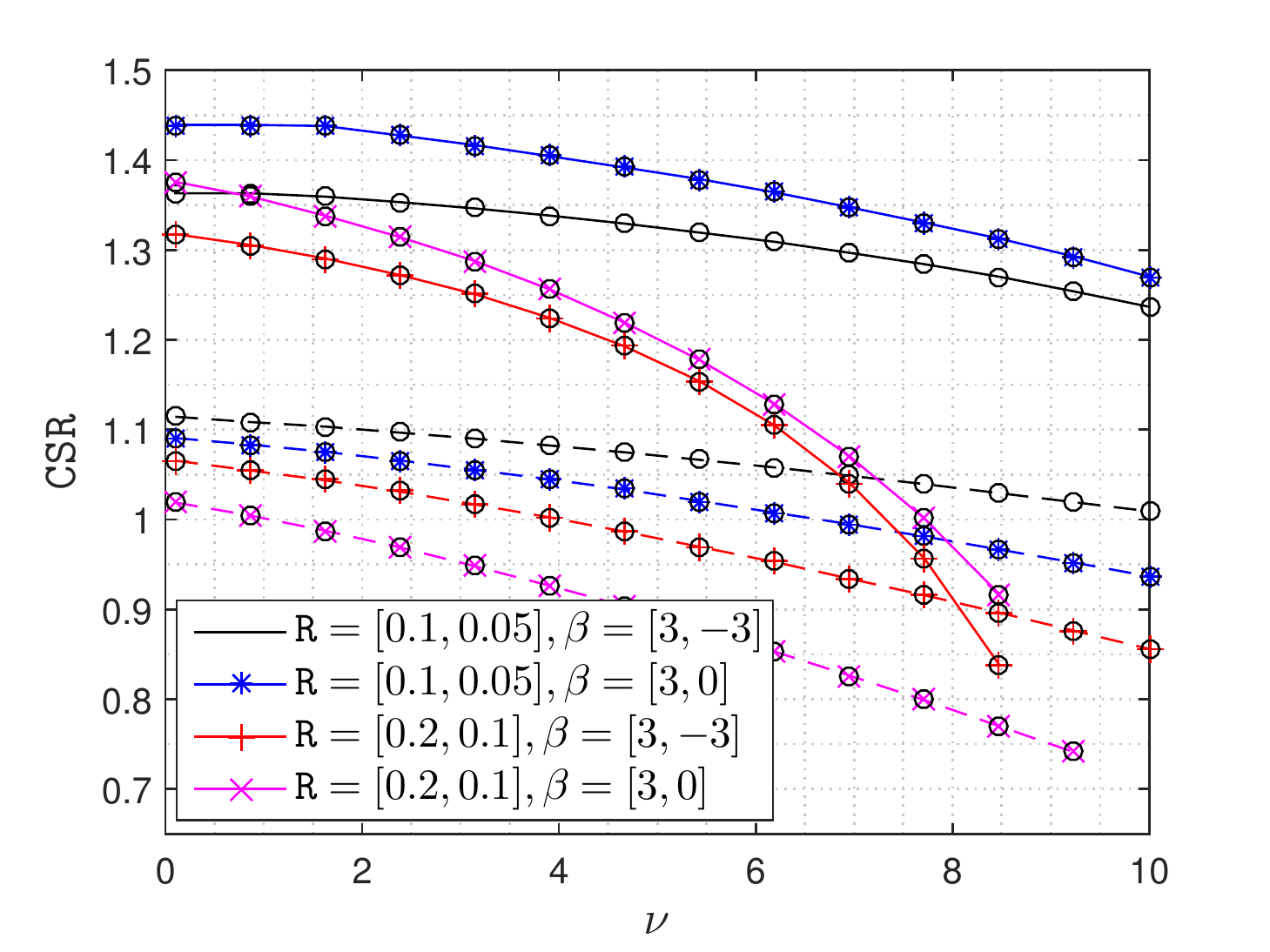} \hspace{-.7cm}
  \includegraphics[width=.35\textwidth]{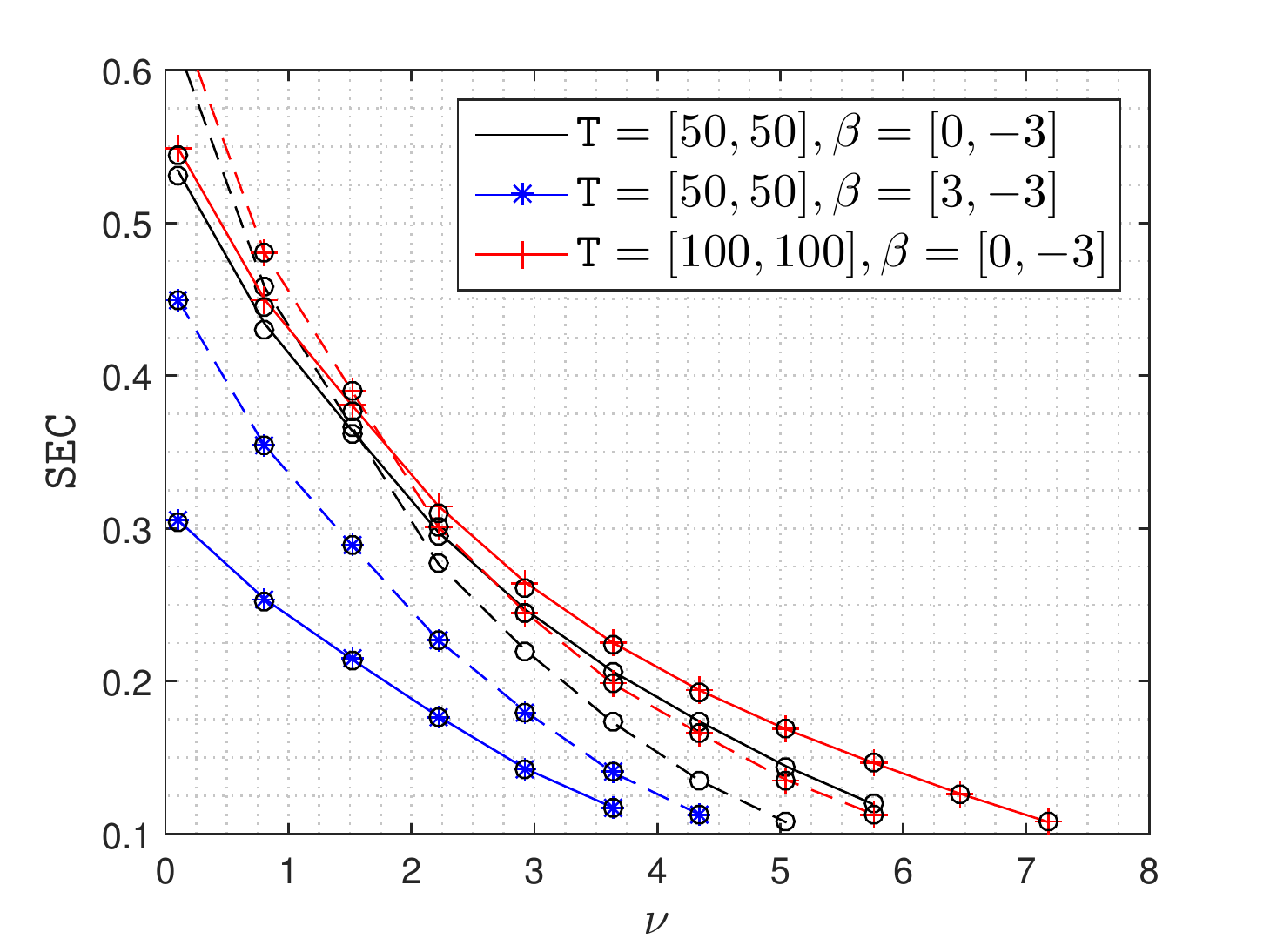}\vspace{-.25cm}
\caption{Rate region of the CC and CE users in NOMA and OMA for $\beta=[\beta_c,\beta_e]=[3, 0]$ in dB (Left).  Maximum $\csr$ under rate  constraints where $\mathtt{R}=[\mathtt{R_c}, \mathtt{R_e}]$ (Middle) and maximum $\sec$ under minimum $\ec$ constraints (Right) for $(\mathtt{O_c},\mathtt{O_e})=(0.2,0.2)$ and $(\bar{\varrho}_c,\bar{\varrho}_e)=(0.05,0.05)$ where $\mathtt{T}=[\mathtt{t_c}, \mathtt{t_e}]$. The solid and dashed lines correspond to the NOMA and OMA, respectively. }
\label{fig:RateRegion_CSEoptimization}\vspace{-3mm}
\end{figure*}
Fig. \ref{fig:Rate_Delay_Outage} verifies that the transmission rate outage probabilities  and upper bounds of the mean delay outage probabilities of the CC and CE users closely match with the simulation results. The outage probabilities  follow the trends opposite  to the  mean transmission rates. 
The outage probabilities of the CC user degrade with {the} increase of $\tau$ because of two reasons: 1) increase in the mean number of CC  users  (i.e., degraded scheduling probability) and 2) decrease in the success probability of CC users. However, a similar direct trend is not visible for the CE user with respect to $\tau$.  {This is because increasing $\tau$ results in the decrease of both the mean number of CE users (i.e., improved scheduling instances) and  success probability for the CE users. }

{ Fig. \ref{fig:RateRegion_CSEoptimization} (Left) illustrates the transmission rate region for the CC and CE users under the NOMA and OMA systems. For OMA, the rate region is the linear combination of maximum transmission rates of the CC and CE users where $\eta=1$ and $\eta=0$ correspond to their maximum transmission rates, respectively. Figure depicts that the performance gains of both the CC and CE users given in \eqref{eq:NOMA_Gain} are above unity for $\eta\leq 0.72$ which implies that the intersection of $\Theta_c(\eta)$ and $\Theta_e(\eta)$ are non-empty for $\eta\in[0,72]$. Further, it can be clearly observed that the NOMA performance gains are independent of $\nu$, whereas the transmission rate regions scale down with the increase of $\nu$. This is because increasing $\nu$ lowers the scheduling probabilities of these users, which as a result, affects their transmission rates but not their relative performance gains.  Besides, the optimal power allocation with respect to the CC users is $\hat\theta$ for the NOMA system. At this point, the CE users receive a non-zero transmission rate as they are allocated with the remaining power of $1-\hat\theta$. Hence, unlike the OMA, the CE users can receive a non-zero transmission rate when the CC users' transmission rate is maximum under the NOMA system. }

{Fig. \ref{fig:RateRegion_CSEoptimization} (Middle) shows the achievable $\csr$ vs. $\nu$ under the minimum transmission rate constraints and Fig \ref{fig:RateRegion_CSEoptimization} (Right) shows the achievable $\sec$ vs. $\nu$ under the minimum $\ec$ constraints. The circular markers correspond to the proposed solutions and the curves correspond to the exact optimal solutions which are obtained through the brute-force search.  We observe that the maximum achievable $\csr$ and $\sec$ drop with the increase of $\nu$ for both NOMA and OMA systems.}  This is because the mean transmission rates of both CC and CE users drop with the increase of $\nu$ (i.e., the increase of the number of users sharing the same RB).  This restricts the feasible range of $\theta$ and thus that of both $\csr$ and $\sec$. The middle figure depicts that the NOMA provides better $\csr$ as compared to OMA. The $\csr$ drops at a higher rate when the minimum required transmission rates are higher as it causes  the power consumption {to increase with the increase in $\nu$ aggressively}.  {Fig \ref{fig:RateRegion_CSEoptimization} (Right)}  depicts that the NOMA provides better $\sec$ for lower $\sir$ threshold at higher values of $\nu$ as compared to OMA. This is because {their scheduling probabilities predominantly determine the packet service rates of CC and CE users (thus the $\sec$)  for the lower values of $\sir$ thresholds.} 
\section{Conclusion}
This paper has provided a comprehensive analysis of downlink two-user NOMA enabled cellular networks for both RT and NRT services. In particular, a new 3GPP-inspired user ranking technique has been proposed wherein the CC and CE users are paired for the non-orthogonal transmission.  
To the best of our knowledge, this is the first stochastic geometry-based approach to analyze downlink NOMA using  a 3GPP-inspired user ranking scheme
that depends upon both the link qualities from the serving and dominant interfering BSs. Unlike the ranking techniques used in the literature, the proposed technique ranks users accurately with distinct link qualities{,} which is vital to obtain  performance gains in NOMA. For the proposed user ranking, we first derive the moments and approximate distributions of the meta distributions for the CC and CE users. Next, we obtain the distributions of their transmission rates and mean delays under random scheduling, which are then used to characterize $\csr$ for NRT service and $\sec$ for RT service, respectively. Using these results, we investigate two RA techniques with objectives to maximize $\csr$ and $\sec$.  
The numerical results demonstrated that NOMA{,} along with the proposed user ranking technique{,} results in a significantly higher $\csr$ and improved transmission rate region for  the CC and CE users as compared to OMA. Besides, our results showed  that NOMA provides improved $\sec$ as compared to OMA for the higher user density. 
The natural extension of this work is to analyze the $N$-user downlink NOMA in cellular networks. Besides, one could also employ the proposed way of user partitioning to analyze the transmission schemes relying on the partition of users based on their perceived link qualities, such as SFR.
\section*{Appendix}
\subsection{Proof of Lemma \ref{lemma:DistanceDistributions}}
\label{app:DistanceDistributions}
Using \eqref{eq:pdf_RoRd} and the definitions of $\Psi_{cc}$ and $\Psi_{ce}$ given in \eqref{eq:CC_CE_pps}, it is easy to see that a uniformly distributed user  in the typical cell $V_o$ is the typical CC user if $R_o\leq \tau R_d$ and the typical CE user if $R_o>\tau R_d$. Therefore, the probability that the typical user is the CC user becomes $\P\left[R_o\leq R_d\tau\right]=$
\begin{align*}
&(2\pi\rho\lambda)^2\int\limits_0^\infty\int\limits_0^{ r_d\tau}r_or_d\exp(-\pi\lambda\rho r_d^2){\rm d}r_o{\rm d}r_d=\tau^2.
\end{align*} 
and thus the probability that the typical user is the CE user becomes $1-\tau^2$.
Using Eq. \eqref{eq:pdf_RoRd} and $\P[R_o\leq R_d\tau]=\tau^2$, we obtain $\cdf$ of $R_o$ for the CC user as $F^{\text{c}}_{R_o}(r_o)=$
\begin{align}
\P[R_o\leq r_o|R_o<R_d\tau]&=\frac{(2\pi\rho\lambda)^2}{\tau^2}\int\limits_0^{r_o}\int\limits_{\frac{u}{\tau}}^\infty u v\exp(-\pi\rho\lambda v^2){\rm d}v{\rm d}u\nonumber\\
&=1-\exp\left(-\pi\rho\lambda{r_o^2}/{\tau^2}\right), 
\label{eq:CDF_Ro_proof}
\end{align} 
for $r_o\geq 0$. Using $R_o\leq R_d\tau$, we obtain the $\cdf$ of $R_d$ of CC user as
\begin{align}
F^{\text{c}}_{R_d\mid R_o}\left(r_d\mid r_o\right)&=\P[R_d\leq r_d\mid R_d>r_o/\tau]\nonumber\\
&=1-\exp\left(-\pi\rho\lambda\left(r_d^2-{r_o^2}/{\tau^2}\right)\right),
\label{eq:CDF_Rd_given_Ro_proof}
\end{align}
$\text{for}~r_d\geq \frac{r_o}{\tau}$. Therefore, the joint $\pdf$ of $R_o$ and $R_d$ for the CC user given in \eqref{eq:pdf_RoRd_CC} directly follows using the Bayes' theorem along with \eqref{eq:CDF_Ro_proof} and  \eqref{eq:CDF_Rd_given_Ro_proof}.
 Similarly, using \eqref{eq:pdf_RoRd} and $\P[R_o> R_d\tau]=1-\tau^2$, we obtain  the $\cdf$ of $R_o$ for the CE user as 
\begin{align}
F^{\text{e}}_{R_o}(r_o)&=\P[R_o\leq r_o\mid R_o>R_d\tau]\nonumber\\
&=\frac{(2\pi\rho\lambda)^2}{1-\tau^2}\int_0^{r_o}\int_{u}^{\frac{u}{\tau}} uv\exp(-\pi\rho\lambda v^2){\rm d}v{\rm d}u\nonumber\\
&=1-\frac{1-\tau^2\exp\left(-\pi\rho\lambda{r_o^2}(\tau^{-2}-1)\right)}{(1-\tau^2)\exp\left(\pi\rho\lambda{r_o^2}\right)}
\label{eq:CDF_Ro_proof_1}
\end{align}
for $r_o>0$. 
Now, the conditional $\cdf$ of $R_d$ given $R_o$ for the CE user can be determined as
\begin{align}
F^{\text{e}}_{R_d\mid R_o}(r_d\mid r_o)&=\P\left[R_d\leq r_d\mid R_d<\frac{r_o}{\tau}\right]\nonumber\\
&=\frac{\P[R_d\leq r_d]}{\P[R_d<\frac{r_o}{\tau}]}\nonumber\\
&=\frac{1-\exp\left(-\pi\rho\lambda(r_d^2-r_o^2)\right)}{1-\exp\left(-\pi\rho\lambda r_o^2(\tau^{-2}-1)\right)},
\label{eq:CDF_Rd_given_Ro_proof_1}
\end{align}
for $r_o<r_d<\frac{r_o}{\tau}$ and $F^{\text{e}}_{R_d\mid R_o}(r_d\mid r_o)=1$ for $r_d\geq\frac{r_o}{\tau}$. 
Finally, the joint $\pdf$ of $R_o$ and $R_d$ for the CE given in Eq. \eqref{eq:pdf_RoRd_CE} directly follows using the Bayes' theorem and Equations \eqref{eq:CDF_Ro_proof_1} and \eqref{eq:CDF_Rd_given_Ro_proof_1}. 
\subsection{Proof of Theorem \ref{thm:MetaDisMoment}}
\label{app:MetaDisMoment}
The success probability of the CC user at $\y\in V_o$ conditioned on $\y=R_o$ and $\Phi$ is 
\begin{align*}
\p_c(\chi_c\mid\y, \Phi)&=\P\left(\ncalE_c\mid\y,\Phi\right)\stackrel{(a)}{=}\prod\limits_{\x\in\Phi}\frac{1}{1+R_o^\alpha \chi_c \|\x-\y\|^{-\alpha}},
\end{align*}
 where step (a) follows from the independence of the fading gains. 
Let $\x_d=\arg\min_{\x\in\Phi}\|\x-\y\|$ and $\tilde{\Phi}=\Phi\setminus\{\x_d\}$. Recall $R_d=\|\x_d\|$. The $b$-th moment of $p_c(\chi_c|\y,\Phi)$ becomes $M^{c}_b(\chi_c)$
\begin{align*}
&=\nbbE_{R_o,\Phi}\left[\prod\limits_{\x\in{\Phi}}\frac{1}{(1+R_o^\alpha \chi_c \|\x-\y\|^{-\alpha})^b}\right],\\
&=\nbbE_{R_o,R_d}\left[\E_{\tilde{\Phi}}\left[\prod\limits_{\x\in\tilde{\Phi}}\frac{(1+\chi_c({R_o}/{R_d})^{\alpha})^{-b}}{(1+R_o^\alpha \chi_c \|\x-\y\|^{-\alpha})^b}\mid \x_d\right]\right],\\
&\stackrel{(a)}{=}\E_{R_o,R_d}\left[\frac{\exp\left(-\lambda\int_{\ncalB_{\y}(R_d)}\left[1-(1+R_o^\alpha \chi_c \|\x-\y\|^{-\alpha})^{-b}\right]{\rm d}\x\right)}{(1+\chi_c({R_o}/{R_d})^{\alpha})^b}\right],\\
&=\E_{R_o,R_d}\left[\frac{\exp\left(-\pi\lambda R_o^2\tilde{\ncalZ}_b\left(\chi_c,{R_o}/{R_d}\right)\right)}{(1+\chi_c({R_o}/{R_d})^{\alpha})^b}\right],
\end{align*}
 where $\ncalB_{\y}^c(r)=\nbbR^2\setminus\ncalB_{\y}(r)$, $\tilde{\ncalZ}_b\left(\chi_c,a\right)=\chi_c^\delta\int_{\chi_c^{-\delta} a^{-2}}^\infty [1-(1+t^{-\frac{1}{\delta}})^{-b}]{\rm d}t,$ and
 step (a) follows by approximating $\tilde{\Phi}$  with the homogeneous PPP with density $\lambda$ outside of the disk $\ncalB_\y\left(R_d\right)$ and the probability generating functional of PPP. Now, using the joint $\pdf$ of $R_o$ and $R_d$ for the CC user given in \eqref{eq:pdf_RoRd_CC}, we get 
\begin{align*}
M_b^{c}(\chi_c)=\frac{(2\pi\rho\lambda)^2}{\tau^2}&\int_0^\infty r_d\exp(-\pi\rho\lambda r_d^2)\times\\
\int_0^{\tau r_d}&\frac{\exp\left(-\pi\lambda r_o^2\tilde\ncalZ_b\left(\chi_c,{r_o}/{r_d}\right)\right)}{(1+\chi_c({r_o}/{r_d})^{\alpha})^b}r_o{\rm d}r_o{\rm d}r_d,\\
\stackrel{(a)}{=}\frac{2(\pi\rho\lambda)^2}{\tau^2}&\int_0^{\tau^2}\frac{1}{(1+\chi_cv^{\frac{1}{\delta}})^b}\times\\
\int_0^\infty  r_d^3&\exp\left(-\pi\lambda r_d^2\left[\rho + \tilde\ncalZ_b\left(\chi_c,v^\frac{1}{2}\right)\right]\right){\rm d}r_d{\rm d}v,
\end{align*} 
where step (a) follows using the substitution $(r_o/r_d)^\alpha=v^\frac{1}{\delta}$ and the exchange of the integral orders.
Further, solving the inner integral gives $M_b^c(\chi_c)$ as in \eqref{eq:MetaDisMoment_CC_CE} such that $\ncalZ(\chi_c,v)=\tilde{\ncalZ}(\chi_c,v^{\frac{1}{2}})$.
Following similar steps and using the joint $\pdf$ of $R_o$ and $R_d$ given in \eqref{eq:pdf_RoRd_CE}, the $b$-th moment of the conditional success probability for the CE user can be obtained as 
\begin{align*}
M_b^{e}(\chi_e)=\frac{(2\pi\rho\lambda)^2}{1-\tau^2}&\int_0^\infty r_d\exp(-\pi\rho\lambda r_d^2)\times\\
\int_{\tau r_d}^{r_d} &\frac{\exp\left(-\pi\lambda r_o^2\tilde\ncalZ_b\left(\chi_e,{r_o}/{r_d}\right)\right)}{(1+\chi_e({r_o}/{r_d})^{\alpha})^b}r_o{\rm d}r_o{\rm d}r_d,\\
=\frac{2(\pi\rho\lambda)^2}{1-\tau^2}&\int_{\tau^2}^1 \frac{1}{(1+\chi_ev^{\frac{1}{\delta}})^b}\times\\
\int_0^\infty r_d^3 &\exp\left(-\pi\lambda r_d^2\left[\rho + v\tilde\ncalZ_b\left(\chi_e,v^\frac{1}{2}\right)\right]\right){\rm d}r_d{\rm d}v,
\end{align*}
Finally, solving the inner integral gives $M_b^e(\chi_e)$  as in \eqref{eq:MetaDisMoment_CC_CE} where $\ncalZ(\chi_e,v)=\tilde{\ncalZ}(\chi_e,v^{\frac{1}{2}})$.
\subsection{Proof of Lemma \ref{lemma:Area_CC_CE_Region}}
\label{app:CCCERegion}
The $n$-th moment of the area of a random set $A\subset\R^2$ can be obtained as \cite{robbins1944}
\begin{figure*}
 \centering
\hspace{-7mm}\includegraphics[width=.27\textwidth]{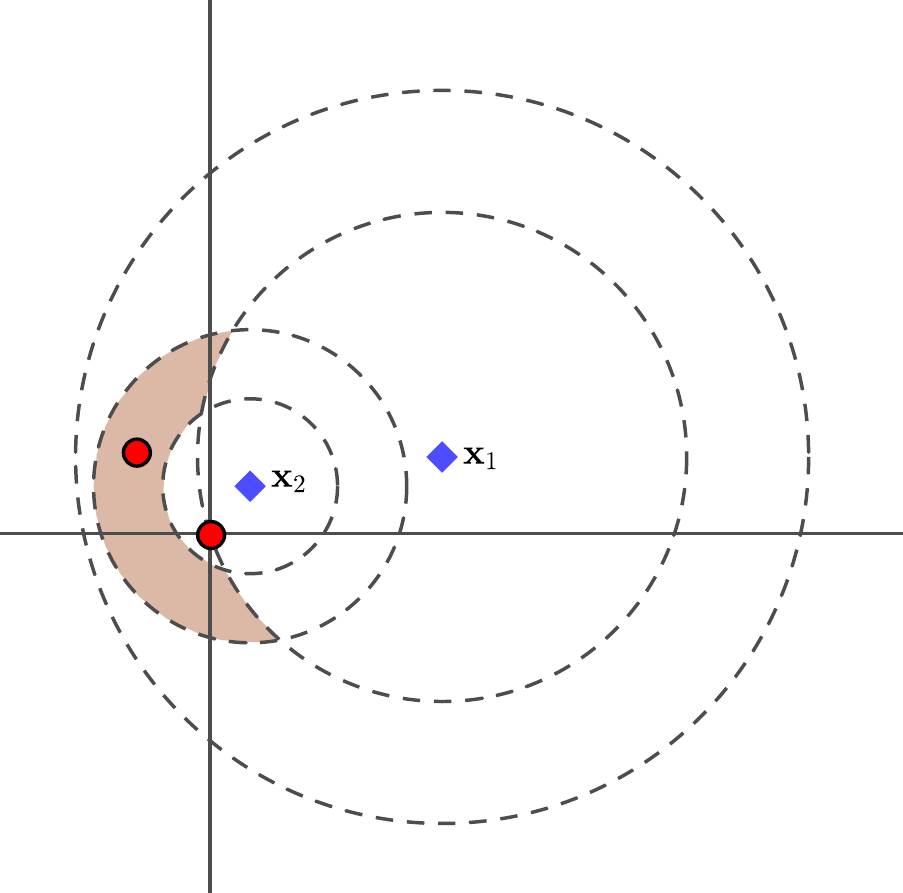} 
\includegraphics[width=.33\textwidth]{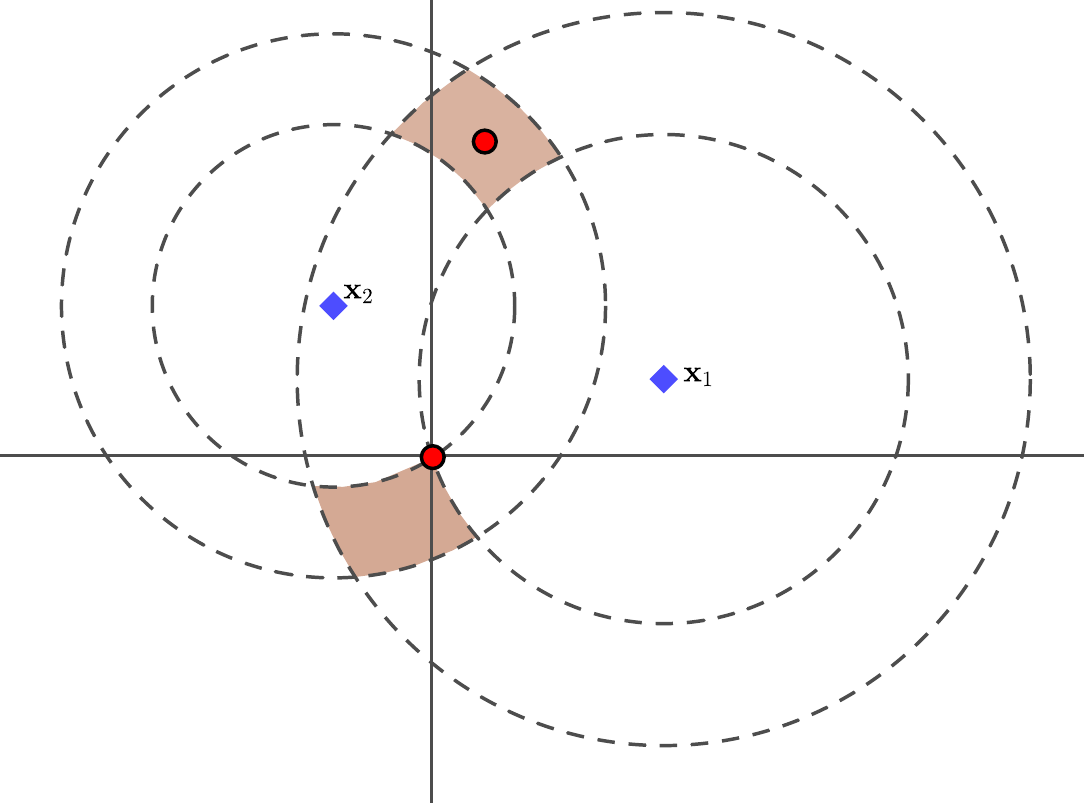} \includegraphics[width=.33\textwidth]{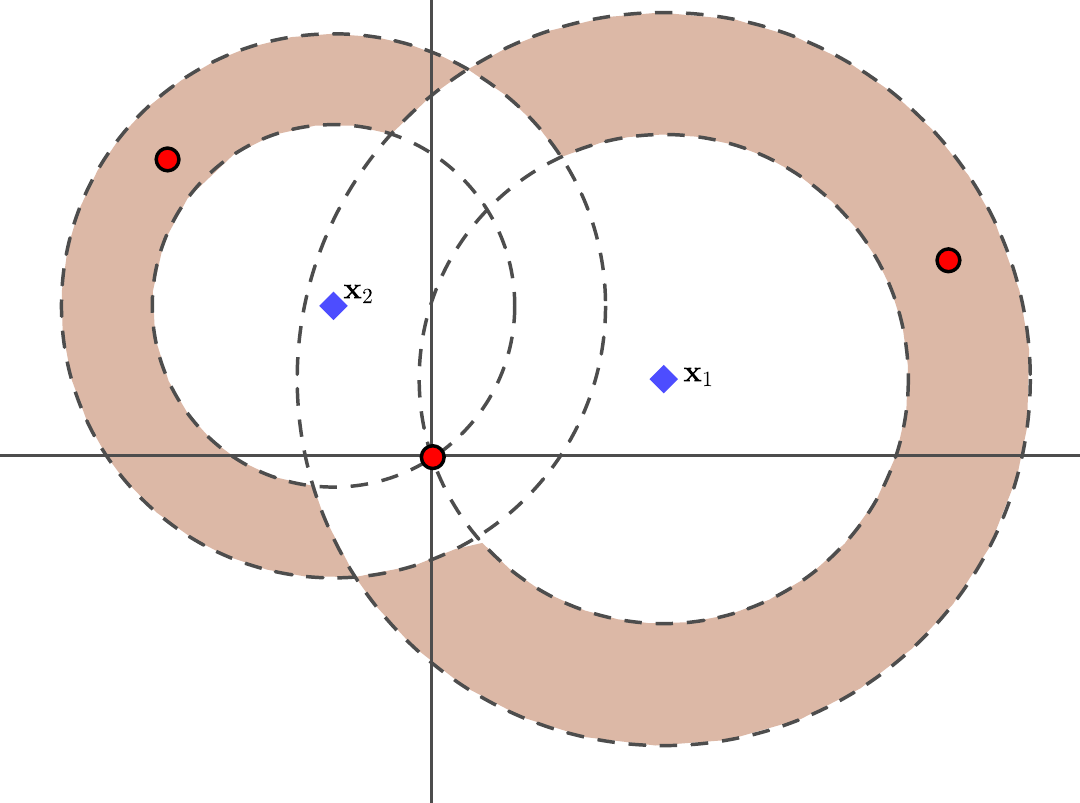}\vspace{-.25cm}\\
~~~~Case 1~~~~~~~~~~~~~~~~~~~~~~~~~~Case 2a~~~~~~~~~~~~~~~~~~~~~~~~~~Case 2b~~~~~~~~~
\caption{Illustration of the cases when $\{\x_1,\x_2\}\in V_{oe}$ given $\{\x_1,\x_2\}\in V_{o}$ (i.e. $\Phi(\ncalC_o)=0$). The blue diamonds represent the locations $\{\x_1,\x_2\}$, whereas the red dots represent the locations of serving and dominant BSs.} 
\label{fig:Cond_X1X2inVoe}
\end{figure*}
\begin{equation}
\mathbb{E}[|A|^n]=\int_{\mathbb{R}^d}\dots\int_{\mathbb{R}^d}\mathbb{P}[\x_1,\dots, \x_n\in A]{\rm d}\x_1\dots{\rm d}\x_n.
\label{eq:Moment_Random_Set}
\end{equation} 
Let $\ncalB_{\x}$ and $\tilde{\ncalB}_{\x}$ be the disks of radii $r$ and ${r}{\tau^{-1}}$ both centered at $\x\equiv(r,\theta)$ and let $\ncalA_{\x}$ be the annulus formed by the two disks $\tilde{\ncalB}_{\x}$ and $\ncalB_{\x}$.
By definition, the  point $\x$ belongs to $V_{oc}$ only if $\Phi({\ncalB}_{\x})=0$ (i.e., $\x\in V_o$) and $\Phi(\ncalA_{\x})=0$. Thus, we have 
\begin{align*}
\nbbP[\x\in V_{oc}]&=\nbbP[\x\in V_{oc}\mid \x\in V_{o}]\nbbP[\x\in V_{o}]\\&\stackrel{(a)}{=}\exp\left(-\lambda|\ncalA_\x|\right)\exp(-\lambda|\ncalB_\x|)\\
&=\exp(-\lambda|\tilde{\ncalB}_\x|),
\label{eq:Prob_YinVoc}\numberthis
\end{align*}
where step (a) follows from the independence property and the {\em void probability} of the PPP.
However, the point $\x$ belongs to $V_{oe}$ only if $\Phi\left(\ncalB_{\x}\right)=0$ and $\Phi\left(\ncalA_{\x}\right)\neq 0$.  Thus, we have
\begin{align*}
\nbbP[\x\in V_{oe}]&=\nbbP[\x\in V_{oe}\mid \x\in V_{o}]\nbbP[\x\in V_{o}]\\
&\stackrel{(a)}{=}\left[1-\exp\left(-\lambda|\ncalA_\x\right)|\right]\exp(-\lambda|\ncalB_\x|)\\
&=\exp(-\lambda|\ncalB_\x|)-\exp(-\lambda|\tilde{\ncalB}_\x|),
\label{eq:Prob_YinVoe}\numberthis
\end{align*}
where step (a) follows from the independence property and the {\em void probability} of the PPP.
Thus, using \eqref{eq:Moment_Random_Set},  \eqref{eq:Prob_YinVoc} and \eqref{eq:Prob_YinVoe},  we get the mean areas of the CC and CE regions as in \eqref{eq:Mean_CCCERegion}. 

Similarly, the probability of $\{\x_1,\x_2\}\in V_{oc}$ can be directly determined as $
\nbbP[\x_1,\x_2\in V_{oc}] = \exp(-\lambda|\ncalC_3|)$, 
where $\ncalC_3=\tilde{\ncalB}_{\x_1}\cup \tilde{\ncalB}_{\x_2}$.
Thus, using this and \eqref{eq:Moment_Random_Set}, we can easily obtain the second moment of the area of the CC region as in \eqref{eq:SecondMoment_CCCERegion}.
Now, we require the probability of $\{\x_1,\x_2\}\in V_{oe}$ for the evaluation of the second moment of area of the CE region. This requires the careful consideration of the intersection of  various sets of two disks. 
Let $d=\|\x_1-\x_2\|=(r_1^2+r_2^2-2r_1r_2\cos(\theta_1-\theta_2))^{\frac{1}{2}}$. Fig. \ref{fig:Cond_X1X2inVoe} shows two cases wherein Case 1 occurs if $d\leq {\tau^{-1}}|r_1-r_2|$, otherwise Case 2 occurs.  Now, we derive $\nbbP[\x_1,\x_2\in V_{oe}]$ for these cases in the following.
\newline{\em Case 1:} In this case, $\{\x_1,\x_2\}\in V_{oe}$ if $\Phi\left(\ncalC_o\right)=0$ (i.e. $\{\x_1,\x_2\}\in V_o$) and if either $\Phi\left(\ncalC_2\setminus\ncalC_o\right)\neq 0\text{~for~}r_2\leq r_1\text{~or~}\Phi\left(\ncalC_1\setminus\ncalC_o\right)\neq 0 \text{~for~}r_1<r_2,$ 
where $\ncalC_o=\ncalB_{\x_1}\cup\ncalB_{\x_2}$, $\ncalC_1=\tilde{\ncalB}_{\x_1}\cup\ncalB_{\x_2}$ and $\ncalC_2=\ncalB_{\x_1}\cup\tilde{\ncalB}_{\x_2}$.
Fig. \ref{fig:Cond_X1X2inVoe} (Left) depicts the second condition of this case for $r_2\leq r_1$. Thus, we get 
\begin{align}
\nbbP[\x_1,\x_2\in V_{o}] &= \exp(-\lambda|\ncalC_o|),\label{eq:Prob_Y1Y2inVoe1}\\
\nbbP[\x_1,\x_2\in V_{oe}\mid \x_1,\x_2\in V_{o}]&= \exp\left(-\lambda(|\ncalC_1|-|\ncalC_o|)\right)\mathbbm{1}_{r_1\leq r_2}\nonumber\\
+ \exp(-\lambda &(|\ncalC_2|-|\ncalC_o|))\mathbbm{1}_{r_2< r_1}.
\label{eq:Prob_Y1Y2inVoe2_Case1}
\end{align}
Therefore, using \eqref{eq:Prob_Y1Y2inVoe1} and \eqref{eq:Prob_Y1Y2inVoe2_Case1} and the independence property of the PPP, we obtain $\nbbP[\x_1,\x_2\in V_{oe}]=$
\begin{align}
 &\exp\left(-\lambda|\ncalC_1|\right)\mathbbm{1}_{r_1\leq r_2} + \exp\left(-\lambda|\ncalC_2|\right)\mathbbm{1}_{r_2< r_1}.\label{eq:Prob_Y1Y2inVoe_Case1}
\end{align}
{\em Case 2:} 
In this case, $\{\x_1,\x_2\}\in V_{oe}$ if $\Phi\left(\ncalC_o\right)=0$ and if one of the following conditions is met: 
2a) $\Phi\left(\ncalA_o\right)\neq 0$, and 2b) $\Phi\left(\ncalC_3\setminus\ncalC_1\right)\neq 0$, and $\Phi\left(\ncalC_3\setminus\ncalC_2\right)\neq 0$,
 where $\ncalA_o=\ncalA_{\x_1}\cap\ncalA_{\x_2}$. The above two cases are depicted in Fig. \ref{fig:Cond_X1X2inVoe} (Middle and Right).
 Thus, $\nbbP[\x_1,\x_2\in V_{oe}\mid \x_1,\x_2\in V_{o}]=$\vspace{-.2cm}
 \begin{align}
&[1-\exp(-\lambda|\ncalA_{o}|)] +\exp(-\lambda|\ncalA_{o}|)\times\nonumber\\
&[1-\exp(-\lambda|\ncalC_3\setminus\ncalC_1|)][1-\exp(-\lambda|\ncalC_3\setminus\ncalC_2|)].\label{eq:Prob_Y1Y2inVoe2_Case2}
\end{align}
We have $|\ncalA_o|+|\ncalC_o|=|\ncalC_1|+|\ncalC_2|-|\ncalC_3|$. Therefore, using \eqref{eq:Prob_Y1Y2inVoe1} and \eqref{eq:Prob_Y1Y2inVoe2_Case2}, we obtain $\nbbP[\x_1,\x_2\in V_{oe}]=$ 
\begin{align}
&[\exp(-\lambda(|\ncalC_2|-|\ncalC_3|))-1][\exp(-\lambda|\ncalC_1|)-\exp(-\lambda|\ncalC_3|)]\nonumber\\
+&[\exp(-\lambda|\ncalC_o|)-\exp(-\lambda(|\ncalC_1|+|\ncalC_2|-|\ncalC_3|))].\label{eq:Prob_Y1Y2inVoe_Case2}
\end{align}
Now, we derive the areas of $\ncalC_o$, $\ncalC_1$, $\ncalC_2$, and $\ncalC_3$. 
Let $U(z_1,z_2,u)$ be the area of the union of two circles of radii $z_1$ and $z_2$ with the angular separation of $u$ between their centers w.r.t.~their intersection points. Thus, $U(z_1,z_2,u)$ can be easily obtained as  given in Lemma \ref{lemma:Area_CC_CE_Region}. 
Without loss of generality, we set $\theta_2=0$ and $0\leq \theta_1<\pi$.
The two disks in $\ncalC_o$, $\ncalC_1$, $\ncalC_2$ and $\ncalC_3$ intersect (if they interest at all) at angles {\small $u_o=\theta_1$, 
\begin{align*}
u_1&=\arccos\left((\tau^{-1}-\tau)\frac{r_1}{2r_2}+\tau\cos(\theta_1)\right),\\ u_2&=\arccos\left((\tau^{-1}-\tau)\frac{r_2}{2r_1}+\tau\cos(\theta_1)\right) \\\text{and~} u_3&=\arccos\left((1-\tau^2)\frac{r_1^2+r_2^2}{2r_1r_2}+\tau^2\cos(\theta_1)\right),
\end{align*}}
respectively.
 The evaluation of $|\ncalC_3|$ is required only for Case 2 wherein $\tilde{\ncalB}_{\x_1}$ and $\tilde{\ncalB}_{\x_2}$ alway intersect. Thus, by definition, we have 
$|\ncalC_o|=U(r_1,r_2,u_o)$  and $|\ncalC_3|=U\left({r_1}{\tau^{-1}},{r_2}{\tau^{-1}},u_3\right)$.
However, $\tilde{\ncalB}_{\x_1}$ and ${\ncalB}_{\x_2}$ (${\ncalB}_{\x_1}$ and $\tilde{\ncalB}_{\x_2}$) intersect only if $\frac{r_1}{\tau}<d+r_2$ ($\frac{r_2}{\tau}<d+r_1$), otherwise $\ncalB_{\x_2}\subset\tilde{\ncalB}_{\x_1}$ ($\ncalB_{\x_1}\subset\tilde{\ncalB}_{\x_2}$).  Thus, we get
{
\begin{align*}
|\ncalC_1| =  
\begin{dcases}
U\left({r_1}{\tau^{-1}},r_2,u_1\right),\hspace{-.2cm} &\text{if~} {r_1}{\tau^{-1}}<d+r_2,\\
\pi{r_1^2}{\tau^{-2}}, &\text{otherwise},
\end{dcases}
\end{align*}
\text{and}
\begin{align*}
|\ncalC_2| =  
\begin{dcases}
U\left(r_1,{r_2}{\tau^{-1}},u_2\right),\hspace{-.2cm}&\text{if~} {r_2}{\tau^{-1}}<d+r_1,\\
\pi{r_2^2}{\tau^{-2}}, &\text{otherwise}.
\end{dcases}
\end{align*}}
Finally, substituting above expressions in \eqref{eq:Prob_Y1Y2inVoe_Case1} and \eqref{eq:Prob_Y1Y2inVoe_Case2}, and then integrating \eqref{eq:Prob_Y1Y2inVoe_Case1} over domain $\{d\leq \tau^{-1}|r_1-r_2|\}$  (i.e., Case 1) and \eqref{eq:Prob_Y1Y2inVoe_Case2} over domain $\{d>\tau^{-1}|r_1-r_2|\}$ (i.e., Case 2), we get the second moment of the area of the CE region as in \eqref{eq:SecondMoment_CCCERegion}.
\subsection{Proof of Theorem \ref{thm:TransmissionRate_NOMA}}
\label{app:TransmissionRate_NOMA}
We first derive the mean transmission rate of the typical CC user at $\y\sim U(V_o)$ as $\bar{R}_c(\chi_c)$
\begin{align*}
&=\E_{\y,\Phi}\left[\frac{1}{N_{oc}}\log_2(1+\beta_c)\E\left[\mathbbm{1}_{\ncalE_c}(\sir_{c},\sir_{e})\mid\y,\Phi\right]\right],\\
&\stackrel{(a)}{=}\log_2(1+\beta_c)\E_{|V_{oc}|}\left[\sum_{n=1}^{\infty}\frac{1}{n}\P(N_{oc}=n\mid |V_{oc}|)\right]\times\\
&~~~~~~~~~~~~~~~~~~~~\E_{\y,\Phi}\left[p_c(\beta_c,\beta_e \mid \y,\Phi)\right],\\
&=\log_2(1+\beta_c)\sum_{n=1}^\infty\frac{1}{n}\P(N_{oc}=n)\E_{\y,\Phi}\left[p_c(\beta_c,\beta_e \mid \y,\Phi)\right],
\end{align*}
where step (a) follows using Assumption \ref{assumption:Independence_PVCellArea_SuccessProb}. From the definition of the meta distribution, we have $\E_{\y,\Phi}[p_c(\beta_c,\beta_e \mid \Phi)]=M_{1}^{c}(\chi_c,1)$ for $q_c=q_e=1$. 
Hence, using  the distribution of CC load given in \eqref{eq:PMF_Noc}, we obtained $\bar{R}_c(\chi_c)$ as in \eqref{eq:MeanTransmissionRate_CC_CE}. Similarly, using the distribution of CE load given in \eqref{eq:PMF_Noe}, we obtained the mean transmission rate $\bar{R}_e(\chi_e)$ of the typical CE user as in \eqref{eq:MeanTransmissionRate_CC_CE}.
Now, we obtain the distribution of the conditional transmission rate of the typical CC user as 
\begin{align*}
&\ncalR_c(\mathtt{r_{c}};\chi_c)=\P\left[p_c(\beta_c,\beta_e\mid\y,\Phi)\leq \mathtt{r_{c}} N_{oc}\log_2(1+\beta_c)^{-1}\right],\\
&\stackrel{(a)}{=}\E_{N_{oc}}\left[\P\left[p_c(\beta_c,\beta_e\mid\y,\Phi)\leq \mathtt{r_{c}} N_{oc}\log_2(1+\beta_c)^{-1}\mid N_{oc}\right]\right],\\
&\stackrel{(b)}{=}\E_{N_{oc}}\left[I\left(\min\left(\mathtt{r_{c}} N_{oc}\log_2(1+\beta_c)^{-1},1\right);\kappa_{1c},\kappa_{2c}\right)\right],
\end{align*}
where step (a) follows using Assumption \ref{assumption:Independence_PVCellArea_SuccessProb}, and step (b) follows using the beta approximation of the meta distribution of the success probability (see \eqref{eq:BetaApp}). Similarly, we obtain the distribution of the conditional transmission rate of the typical CE user as in \eqref{eq:RateCDF_UpperBound_CE}. 

\begin{IEEEbiography} [{\includegraphics[width=1in,height=1.25in,clip,keepaspectratio]{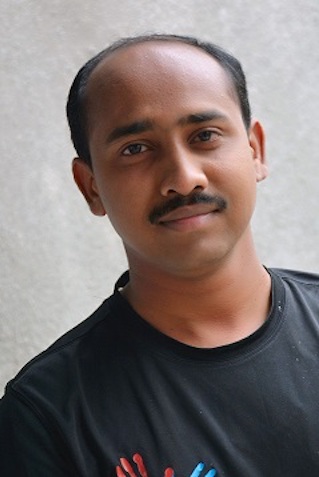}}]{Praful D. Mankar}	
(S'14--M'19) Praful received his B.E. degree in Electrical and Communication Engineering from Amravati University, MH, India, in 2006, and the M.Tech degree in Telecommunications Systems and Ph.D. degree in wireless networks from IIT Kharagpur, WB, India, in 2009 and 2016, respectively. 

He worked as a senior research fellow from Aug. 2009 to March 2015 and a research assistant from Jan. 2016 to Aug. 2017 at the G. S. Sanyal School of Telecommunications (G.S.S.S.T), IIT Kharagpur. 
From Sept. 2017 to June 2019, he worked as a postdoctoral research associate in Prof. Harpreet Dhillon's research group at the Department of Electrical and Computer Engineering at Virginia Tech, USA. His research interests primarily focus on the modeling and analysis of wireless networks using the tools of stochastic geometry.
\end{IEEEbiography}

\begin{IEEEbiography} [{\includegraphics[width=1in,height=1.25in,clip,keepaspectratio]{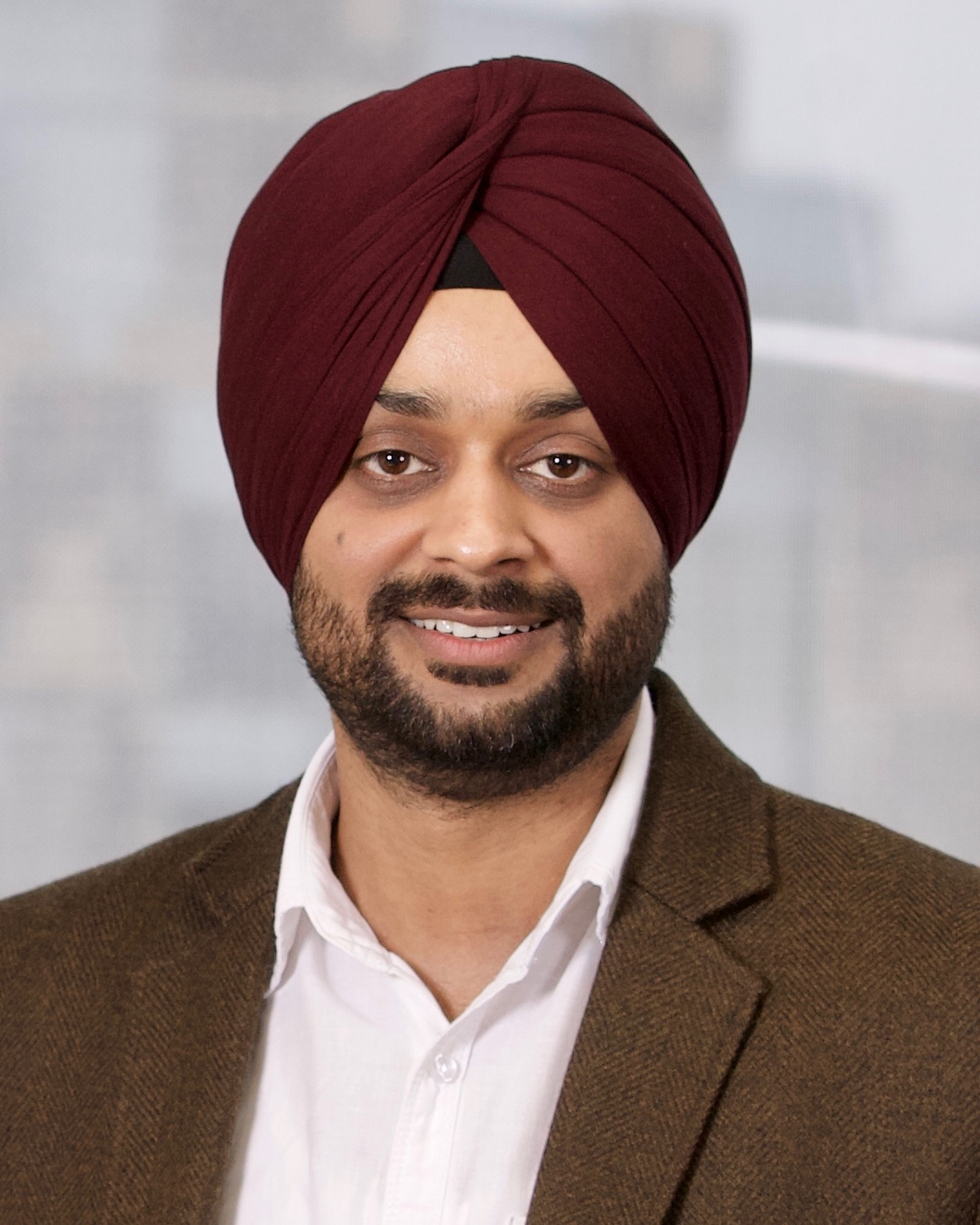}}]{Harpreet S. Dhillon}	
(S'11--M'13--SM'19) received the B.Tech. degree in electronics and communication engineering from IIT Guwahati in 2008, the M.S. degree in electrical engineering from Virginia Tech in 2010, and the Ph.D. degree in electrical engineering from the University of Texas at Austin in 2013. 

After serving as a Viterbi Postdoctoral Fellow at the University of Southern California for a year, he joined Virginia Tech in 2014, where he is currently an Associate Professor of electrical and computer engineering and the Elizabeth and James E. Turner Jr. '56 Faculty Fellow. His research interests include communication theory, wireless networks, stochastic geometry, and machine learning. He is a Clarivate Analytics Highly Cited Researcher and has coauthored five best paper award recipients including the 2014 IEEE Leonard G. Abraham Prize, the 2015 IEEE ComSoc Young Author Best Paper Award, and the 2016 IEEE Heinrich Hertz Award. He was named the 2017 Outstanding New Assistant Professor, the 2018 Steven O. Lane Junior Faculty Fellow, and the 2018 College of Engineering Faculty Fellow by Virginia Tech. His other academic honors include the 2008 Agilent Engineering and Technology Award, the UT Austin MCD Fellowship, and the 2013 UT Austin Wireless Networking and Communications Group leadership award. He currently serves as a Senior Editor for the {\sc IEEE Wireless Communications Letters} and an Editor for the {\sc IEEE Transactions on Wireless Communications} and the {\sc IEEE Transactions on Green Communications and Networking}.
\end{IEEEbiography}
\end{document}